\documentclass[aps,pra,twocolumn,superscriptaddress,nofootinbib]{revtex4-2}

\usepackage[utf8]{inputenc}
\usepackage{amsmath,amssymb,amsthm,mathtools}
\usepackage{braket}
\usepackage{graphicx}
\usepackage{booktabs}
\usepackage{times,color}
\usepackage{tikz}
\usetikzlibrary{arrows.meta,shapes.geometric}
\usepackage[hyperindex,breaklinks,colorlinks=true,linkcolor=black,citecolor=blue,urlcolor=blue]{hyperref}
\usepackage[capitalize,nameinlink]{cleveref}

\newtheorem{theorem}{Theorem}[section]
\newtheorem{lemma}[theorem]{Lemma}

\newtheorem{proposition}[theorem]{Proposition}
\theoremstyle{definition}

\theoremstyle{remark}

\newcommand{\E}{\mathbb E}
\newcommand{\Prob}{\mathbb P}
\newcommand{\supp}{\operatorname{supp}}

\newcommand{\MS}{\mathrm{MS}}
\newcommand{\Dset}{\mathfrak D}
\newcommand{\Dbin}{D_{\mathrm{bin}}}
\newcommand{\pos}[1]{\left[#1\right]_{+}}
\newcommand{\Cb}{\bar C}

\newcommand{\Cexp}{C_{\mathrm{exp}}}
\newcommand{\Cwc}{C_{\mathrm{wc}}}

\begin{document}

\title{Bounding Two-Way Average Communication Cost of Simulating Quantum Correlations}

\author{Kai-Siang Chen}
\affiliation{Department of Physics and Center for Quantum Frontiers of Research \& Technology (QFort), National Cheng Kung University, Tainan 701, Taiwan}

\author{Gelo Noel M. Tabia}
\affiliation{Hon Hai (Foxconn) Research Institute, Taipei, Taiwan}
\affiliation{Department of Physics and Center for Quantum Frontiers of Research \& Technology (QFort), National Cheng Kung University, Tainan 701, Taiwan}
\affiliation{Physics Division, National Center for Theoretical Sciences, Taipei 10617, Taiwan}

\author{Bo-An Tsai}
\affiliation{Department of Physics and Center for Quantum Frontiers of Research \& Technology (QFort), National Cheng Kung University, Tainan 701, Taiwan}

\author{Swati Kumari}
\affiliation{Department of Physics, Patna Science College, Patna University, Bihar 800005, India}
\affiliation{Department of Physics and Center for Quantum Frontiers of Research \& Technology (QFort), National Cheng Kung University, Tainan 701, Taiwan}

\author{Yeong-Cherng Liang}
\email{ycliang@mail.ncku.edu.tw}
\affiliation{Department of Physics and Center for Quantum Frontiers of Research \& Technology (QFort), National Cheng Kung University, Tainan 701, Taiwan}
\affiliation{Physics Division, National Center for Theoretical Sciences, Taipei 10617, Taiwan}

\begin{abstract}
Bell nonlocal correlations cannot be reproduced by local hidden-variable models without communication, making classical communication cost a natural quantitative measure of nonlocality. Although finite communication always suffices to simulate any correlation in a fixed finite Bell scenario, determining the minimum amount needed to (exactly) simulate any given nonlocal correlation remains challenging, especially in the average-cost setting. Here, we derive lower bounds on input-average communication cost from nonlocal games, allowing variable-length, fully interactive two-way protocols. When applied to parallel games, these bounds yield explicit finite quantum correlations for which the exact simulation cost exceeds any prescribed finite input-average communication budget. For $n$ parallel Magic-Square-winning correlations, we prove a lower bound of $n\log_2(3/2)$ bits and give a one-way communication protocol that attains this rate asymptotically. For the quantum correlation maximally winning the $n$-copy of the Clauser--Horne--Shimony--Holt nonlocal game, we obtain an input-average lower bound of approximately $0.04627n$ bits for exact simulation. Finally, we formulate a hierarchy of lower bounds and complementary upper bounds using deterministic correlations up to chosen communication costs. Both recover the exact input-average cost at their final levels, and if the bounds agree at lower levels, their common value yields the exact input-average cost without needing to consider all communication strategies.
\end{abstract}
\date{\today}
\maketitle

\section{Introduction}

Quantum correlations arising from entanglement cannot, in general, be reproduced by classical models without communication~\cite{Bell64,Brunner_RevModPhys_2014}. Quantifying the minimum classical communication required for their simulation therefore provides a natural way to characterize the strength of nonlocality. This question has been studied since the early works on communication-assisted local models~\cite{Maudlin92,Brassard1999,STEINER00,Cerf:PRL:00,Massar:PRA:01}. Remarkably, very little communication suffices in some important cases. For instance, a single classical bit is enough to reproduce exactly the correlations of a maximally entangled two-qubit state under arbitrary projective measurements~\cite{TonerBacon2003}.

A broad range of simulation protocols and communication lower bounds has since been developed. These include exact protocols with finite worst-case or average communication for various entangled states and measurement classes~\cite{Pironio:PRA:2003,Methot:EPJD:2004,Renner:PRL:23,Renner:Quantum:2023}, protocols for reproducing the joint correlators of dichotomic measurements~\cite{Degorre:PRA:2007,Regev10}, and approximate one-bit simulations using numerical and semianalytical methods~\cite{Sidajaya:Quantum:23}. Lower bounds on communication cost have also been derived from quantum violations of inequalities ~\cite{Pironio:PRA:2003,Barrett06,Vertesi:PRA:2009,Marton24,Sidajaya:PRA:2024}, while complementary studies have examined communication-assisted Bell scenarios through causal models~\cite{Brask:17,Vieira2025} or their corresponding correlation polytopes~\cite{BaconToner_2003,Maxwell_2014,ZC:Entropy:2019}. Meanwhile, the nonlocal-capacity framework characterizes the asymptotic cost of exact one-way simulation and gives finite-copy bounds under an entropic communication measure ~\cite{MontinaWolf:NJP:2016}. 
Related questions have also been studied in multipartite scenarios~\cite{Tessier:PRA:2005,Barrett:PRA:2007,Broadbent09,Bancal:AMP:10,Branciard:PRL:11,Branciard:PRL:2012,Brassard:IEEETIT:2016,Brassard:Entropy:2019}.

A complementary line of work uses communication-complexity
techniques~\cite{Buhrman:RMP:2010} and parallel-repetition
arguments~\cite{Raz1998Parallel,Holenstein2009,Rao2011}
to construct families of correlations whose classical simulation cost
grows with the size of the problem. Using the distributed
Deutsch--Jozsa problem~\cite{Buhrman98,DeutschJozsa92}, Brassard, Cleve, and Tapp showed that exact
classical simulation of certain correlations obtained from $k$ Bell
pairs requires $\Omega(2^k)$ bits for sufficiently large
$k$~\cite{Brassard1999}. The same underlying communication problem
also admits an average-communication lower bound with the same
asymptotic scaling for a suitable input
distribution~\cite{Buhrman98}. Nonlocal Hidden Matching provides
other families with growing one-way communication
cost~\cite{BarYossef:STOC:2004,Buhrman:TC:2012}. For parallel Magic
Square (MS) games~\cite{Mermin_90,Peres1990,Cabello01,Aravind02}, Jain and Kundu~\cite{JainKundu2025} (see also~\cite{HasegawaLeGallModanese2025})
established a bound on the fraction of games that can be won in terms of the amount of communication.
These results establish growing communication requirements, with the
focus primarily on asymptotic scaling. Such asymptotic results, however, do not by themselves guarantee that (large) communication is required for small
finite instances, as illustrated by the Deutsch--Jozsa
family~\cite{Brassard05}.

Explicit finite separations are also known. M\'arton \emph{et al.}
showed that two copies of the MS game rule out one bit of
worst-case communication~\cite{Marton24}. They further used parallel
Clauser--Horne--Shimony--Holt (CHSH) games
~\cite{CHSH_1969,Barrett:PRA:2002:Parallel} to construct examples excluding
prescribed worst-case communication budgets in a one-way model, with
an extension to a restricted two-way setting. Smaller finite Bell scenarios separating quantum correlations from one-way one-bit protocols have also been found~\cite{Sidajaya:PRA:2024}. These results motivate a related question for input-average communication. Given an arbitrary finite communication budget, can one exhibit a finite quantum correlation whose exact simulation necessarily exceeds that budget even when variable-length, fully interactive two-way protocols are allowed?

Here, we answer this question affirmatively by adapting communication-complexity techniques to the input-average setting. For the simulation of a perfect winning probability, we combine the standard rectangle method~\cite{KushilevitzNisan1997} with an entropy
bound for average communication~\cite{DietzfelbingerWunderlich2007} to
relate the input-average cost to the largest winning-rectangle probability. When perfect winning is relaxed, we instead use the
information about the inputs revealed by the exchanged
messages~\cite{ChakrabartiKondapallyWang2012,BravermanKo2018} to bound the
communication cost. We then extend both arguments to parallel games. For
perfect-winning games, we prove that the largest winning-rectangle probability is multiplicative under parallel repetition. In the non-perfect-winning case, we show
that the parallel bound can instead be related to a single-copy quantity.
Applied respectively to the MS and CHSH games, these results give linear input-average lower bounds and hence explicit finite-copy correlations whose exact simulation exceeds any prescribed finite input-average communication budget. For the MS game, we further construct a one-way simulation protocol achieving the same asymptotic rate.

For a specified finite correlation, the exact input-average cost
can alternatively be written as an optimization over convex
decompositions into correlations generated by deterministic
communication protocols~\cite{Pironio:PRA:2003}. However, storing
all such correlations quickly becomes prohibitively memory-demanding as the Bell scenario grows.
To this end, we introduce a linear programming (LP) hierarchy of lower bounds that includes only deterministic correlations
whose minimum communication costs do not exceed a chosen threshold. Complementarily, upper bounds can be obtained by optimizing over exact decompositions using deterministic correlations with up to some threshold
communication costs.
When the bounds agree, they determine the exact input-average
communication cost of the target correlation.

\section{Preliminaries}
\label{sec:preliminaries}
We consider a bipartite Bell scenario in which Alice and Bob
receive inputs $x\in X$ and $y\in Y$ and produce outputs
$a\in A$ and $b\in B$. We write
$\vec{P}:=\{P(a,b|x,y)\}_{a,b,x,y}$ for the full correlation
and $P(a,b|x,y)$ for its individual components.
A Bell-local correlation has the form
\begin{equation}
 P(a,b|x,y) = \sum_\lambda p(\lambda) P(a|x,\lambda)P(b|y,\lambda),
 \label{eq:local-main}
\end{equation}
where $\lambda\in\Lambda$ denotes (input-independent) classical shared randomness. Then, any correlation that cannot be explained by~\cref{eq:local-main} is said to be nonlocal. Quantum correlations are obtained by local measurements
on a shared bipartite state,
\begin{equation}
 P(a,b|x,y)
 =
 \operatorname{tr}\left[
 \rho_{AB}(M_{a|x}\otimes N_{b|y})
 \right],
\end{equation}
where $\{M_{a|x}\}_a$ and $\{N_{b|y}\}_b$ are the local
positive operator-valued measures for inputs $x$ and $y$, respectively. Both Bell-local and quantum correlations obey the
no-signaling conditions~\cite{Barrett:PRA:2005}, which require each party's marginal
output distribution to be independent of the other party's input:
\begin{subequations}\label{Eq:NS}
\begin{align}
	 \sum_b P(a,b|x,y) &= \sum_b P(a,b|x,y') &&\forall\,\,a,x,y,y',\label{Eq:NS:BtoA}\\
	 \sum_a P(a,b|x,y) &= \sum_a P(a,b|x',y)  &&\forall\,\,b,y,x,x'\label{Eq:NS:AtoB}.
\end{align}
\end{subequations}

In a classical simulation with communication, before producing their outputs, Alice and Bob may exchange messages regulated by shared randomness $\Lambda$ sampled independently of the inputs, with realization $\lambda$ and distribution $p_\Lambda$. 
A protocol $\mathcal R$ specifies how the parties use $\lambda$, communicate, and generate their outputs. Communication may be fully interactive, with messages exchanged in either direction over several rounds. Throughout, we use $\sigma\in\{\to,\leftrightarrow\}$ to denote the allowed communication model, where $\to$ denotes one-way communication from Alice to Bob and $\leftrightarrow$ denotes fully interactive two-way communication.

To describe how a protocol uses communication, we first consider its behavior for a fixed value $\Lambda=\lambda$. Without loss of generality, we may then take the protocol to be deterministic and assign a unique output pair to every input pair.
We denote the resulting correlation by $\vec{D}$, whose
components $D(a,b|x,y)$ are either zero or one. Since
communication is allowed, $\vec{D}$ can be nonlocal and even violate the no-signaling conditions of \cref{Eq:NS}.

For each fixed $\lambda$, we represent the resulting deterministic protocol by a finite rooted public tree known to both parties, as illustrated in
\cref{fig:public-tree-example}. Each internal node $v$ specifies
a sender, either Alice or Bob, and a message alphabet of size
$m(v)$. At each iteration, the sender chooses a symbol based on their {\em local input}
and the public transcript, i.e., the sequence of messages already
exchanged. This symbol determines which outgoing edge is followed.
Each leaf marks the end of communication.

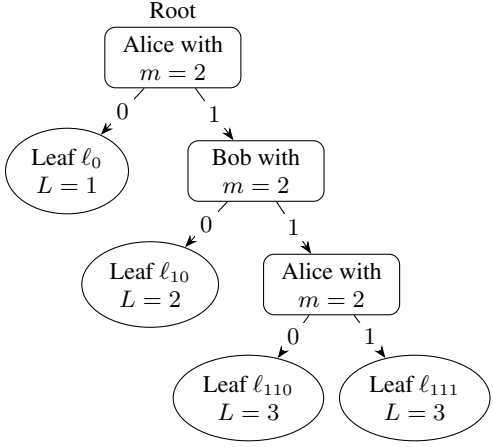
\begin{figure}[t]
\centering
\begin{tikzpicture}[
    >=Stealth,
    internal/.style={
        draw, rounded corners, align=center,
        minimum width=1.8cm, minimum height=0.8cm
    },
    terminal/.style={
        draw, ellipse, align=center,
        minimum width=1.6cm, minimum height=0.8cm
    },
    edge label/.style={
        fill=white, inner sep=2pt, font=\small
    }
]
\node[internal, label=above:{Root}] (root) at (0,0)
    {Alice with\\$m=2$};

\node[terminal] (leaf0) at (-1.4,-1.5)
    {Leaf $\ell_0$\\$L=1$};

\node[internal] (bob) at (1.1,-1.5)
    {Bob with\\$m=2$};

\node[terminal] (leaf10) at (-0.3,-3)
    {Leaf $\ell_{10}$\\$L=2$};

\node[internal] (alice) at (2.1,-3)
    {Alice with\\$m=2$};

\node[terminal] (leaf110) at (1,-4.5)
    {Leaf $\ell_{110}$\\$L=3$};

\node[terminal] (leaf111) at (3.2,-4.5)
    {Leaf $\ell_{111}$\\$L=3$};

\draw[->] (root) -- node[edge label, midway] {$0$} (leaf0);
\draw[->] (root) -- node[edge label, midway] {$1$} (bob);

\draw[->] (bob) -- node[edge label, midway] {$0$} (leaf10);
\draw[->] (bob) -- node[edge label, midway] {$1$} (alice);

\draw[->] (alice) -- node[edge label, midway] {$0$} (leaf110);
\draw[->] (alice) -- node[edge label, midway] {$1$} (leaf111);
\end{tikzpicture}
\caption{An example of a rooted public tree. Upon learning her local input $x$, Alice begins by
sending one bit. Sending $0$ ends communication at $\ell_0$,
whereas sending $1$ passes the turn to Bob. Based on his local input $y$, Bob replies accordingly, which either
ends communication at $\ell_{10}$ when he sends $0$, or returns
the turn to Alice when he sends $1$. In the latter case, Alice's
final bit selects between $\ell_{110}$ and $\ell_{111}$.
The leaf subscripts record the complete message sequences.
Since every internal node has $m(v)=2$, the four paths leading to $\ell_0$, $\ell_{10}$, $\ell_{110}$, and $\ell_{111}$ have
communication costs of $1$, $2$, $3$, and $3$ bits, respectively.}
\label{fig:public-tree-example}
\end{figure}

We assign a cost of $\log_2m$ bits to a message with $m$
possible symbols. A root-to-leaf path $\ell$ therefore incurs the cost:
\begin{equation}
	 L(\ell) := \sum_{v\in\operatorname{path}(\ell)} \log_2m(v),
 \label{eq:path-main}
\end{equation}
where the sum runs over the internal nodes along the path.

For binary protocols, $m(v)=2$ at every internal node, so
$L(\ell)$ equals the number of transmitted bits. For a
general $m$-valued message, our convention assigns a cost
of $\log_2m$, while a fixed-length binary encoding requires
$\lceil\log_2m\rceil$ bits. Throughout, we remove unused branches and nodes with only one child, and assume that the randomness takes finitely many possible values. These restrictions do not affect the minimum communication costs defined below, as shown in Appendix~\ref{app:costs}.

For a protocol $\mathcal R$, let $L_{\mathcal R}(x,y,\lambda)$ be the communication cost given inputs $(x,y)$ and when $\Lambda=\lambda$, i.e., $L_{\mathcal R}(x,y,\lambda)$ is the $L(\ell)$ for the path $\ell$ corresponding to the given $(x,y,\lambda)$. Given an input distribution $\mu$ that assigns positive probability to every input pair, we define the input-average communication cost of simulating
$\vec{P}$ with allowed communication model $\sigma$ by
\begin{equation}
 \Cb_{\mu}^{\sigma}(\vec{P}) :=
 \min_{\mathcal R^{\sigma}\to\vec{P}}
 \E_{\mu,p_\Lambda}
 \bigl[L_{\mathcal R}(x,y,\lambda)\bigr],
 \label{eq:input-average-cost}
\end{equation}
where the minimum is over protocols with communication model $\sigma$ that reproduce $\vec{P}$ exactly on every input pair,
as indicated by $\mathcal R^{\sigma}\to\vec{P}$. The expectation, denoted by $\E$,
averages over inputs drawn from $\mu$ and random variables $\Lambda$
drawn from $p_\Lambda$. Throughout, expectation subscripts specify the distributions over which the average is taken.

Instead of computing~\cref{eq:input-average-cost} directly, we may obtain a lower bound by associating the target correlation $\vec{P}$ with a nonlocal game and asking how much communication is required to attain the winning probability associated with $\vec{P}$.
In a nonlocal game~\cite{Cleve04} $G$, Alice and Bob receive inputs drawn
from a distribution $\mu$ and return outputs subject to a
winning condition. Let $V(a,b,x,y)$ equal one when the
outputs $(a,b)$ win on inputs $(x,y)$ and zero otherwise.
The winning probability of a correlation $\vec{P}$ therefore is
\begin{equation}
 w_G(\vec{P})
 :=
 \sum_{a,b,x,y}
 \mu(x,y)V(a,b,x,y)P(a,b|x,y).
 \label{eq:game-winning-probability-main}
\end{equation}
Within this framework, pseudo-telepathy games~\cite{Brassard05} provide a natural setting for studying quantum-classical separations. They admit quantum strategies that win with certainty, whereas every local strategy loses with nonzero probability.

For a game $G$ with input distribution $\mu$ and required
winning probability $0<\omega\leq1$, let $\mathfrak R^{\sigma}$ denote
the class of protocols allowed by communication model $\sigma$
and win $G$ with probability at least $\omega$, averaged over
$\mu$ and $p_\Lambda$. The minimum input-average communication cost required to achieve this winning probability is
\begin{equation}
 \Cb_\mu^{\sigma,(\omega)}(G)
 :=
 \min_{\mathcal R\in\mathfrak R^{\sigma}}
 \E_{\mu,p_\Lambda}
 [L_{\mathcal R}(x,y,\lambda)],
 \label{eq:game-average-cost-main}
\end{equation}
while the corresponding minimum worst-case communication cost is
\begin{equation}
 C_{\mathrm{wc},\mu}^{\sigma,(\omega)}(G)
 :=
 \min_{\mathcal R\in\mathfrak R^{\sigma}}
 \lVert L_{\mathcal R}\rVert_{\infty,\mu},
 \label{eq:game-wc-cost-main}
\end{equation}
where
\begin{equation}
 \lVert L_{\mathcal R}\rVert_{\infty,\mu}
 :=
 \max_{\substack{(x,y)\in\supp(\mu)\\
 \lambda\in\supp(\Lambda)}}
 L_{\mathcal R}(x,y,\lambda)
\end{equation}
is the largest communication cost over all executions with
inputs in $\supp(\mu)$.

For perfect winning, we omit the winning-probability superscript
and write $\Cb_\mu^\sigma(G):=\Cb_\mu^{\sigma,(1)}(G)$ and
$C_{\mathrm{wc},\mu}^\sigma(G)
:=C_{\mathrm{wc},\mu}^{\sigma,(1)}(G)$.
Since winning a game with certainty does not require reproducing
a particular output distribution, any exact simulator
of a correlation $\vec{P}$ that wins $G$ perfectly also
provides a perfect strategy for the game, giving
\begin{equation}
 \Cb_\mu^\sigma(\vec{P})
 \geq
 \Cb_\mu^\sigma(G)
\end{equation}
and
\begin{equation}
 \Cwc^\sigma(\vec{P})
 \geq
 C_{\mathrm{wc},\mu}^\sigma(G),
\end{equation}
where
\begin{equation}
 \Cwc^\sigma(\vec{P})
 :=
 \min_{\mathcal R^{\sigma}\to\vec{P}}
 \max_{x,y,\lambda}
 L_{\mathcal R^\sigma}(x,y,\lambda).
 \label{eq:worst-case-cost}
\end{equation}

For each fixed $\sigma\in\{\to,\leftrightarrow\}$, since the average communication of any protocol is bounded by its maximum
 communication over all inputs and $\Lambda$, we have
\begin{equation}
 \Cb_\mu^\sigma(\vec{P})
 \leq
 \Cwc^\sigma(\vec{P}).
 \label{eq:cost-chain-main}
\end{equation}
Consequently, every input-average communication lower bound also applies to the corresponding worst-case cost. Likewise, since one-way communication protocols form a special case of fully interactive two-way communication protocols, every two-way lower bound also applies to one-way communication.

\section{Communication bounds from nonlocal games}
\label{sec:game-bounds-main}

\subsection{Perfect winning and winning rectangles}
\label{sec:winning-rectangles-main}

Consider a game in which every input pair in $X\times Y$ is allowed, and $\mu$ has full support, i.e., $\mu(x,y)\neq 0$ for all $x,y$. For every fixed value $\lambda$ of the protocol randomness, the protocol is deterministic, so each input pair $(x,y)$ follows a definite transcript to a leaf, cf.~\cref{fig:public-tree-example}. Let us group together all input pairs that reach the same leaf. Along the transcript leading to that leaf, each message from Alice helps Bob identify her possible inputs, while each message from Bob helps Alice identify his. Hence, the set of input pairs $\{(x,y)\}$ leading to any leaf forms a ``rectangle" $S\times T$, defined by the Cartesian product of all of Alice's inputs $x\in S$ and all of Bob's inputs $y\in T$ compatible with the transcript specific to the leaf.
Once $\lambda$ and the leaf are fixed, each party's final output within this rectangle depends only on its own input.

A rectangle $S\times T\subseteq X\times Y$ is called a
winning rectangle if there exist local response functions
$a:S\to A$ and $b:T\to B$ that win the game on every input
pair in $S\times T$. If the protocol wins with certainty,
the local output functions associated with each leaf satisfy
this condition, so every leaf corresponds to a winning
rectangle. See~\cref{fig:chsh-winning-rectangles} for an
illustration. Let us now define
\begin{equation}
 \alpha
 :=
 \max_{S\times T\ {\rm winning}}\mu(S\times T),
 \label{eq:alpha-main}
\end{equation}
where the maximization is carried out over all winning rectangles found in any rooted public tree resulting from all communication protocols. Thus, $\alpha$
is the largest probability of a winning rectangle under $\mu$.
For uniform inputs, it is simply the largest fraction of
input pairs covered by such a rectangle.

\begin{figure}[t]
\centering
\begin{tikzpicture}[x=1cm,y=1cm,font=\small,
  edge/.style={-{Stealth[length=2.2mm]},line width=.75pt},
  selected/.style={line width=1.1pt},
  every node/.style={inner sep=2pt}]

\node[draw,line width=.75pt,align=center,minimum width=1.3cm,
  minimum height=1.0cm] (alice) at (0,.4) {Alice\\sends $x$};
\draw[edge,color=blue] (alice.east)--(1.8,2.5)
  node[midway,above left] {$x=0$};  
\draw[edge,color=red] (alice.east)--(1.8,-1.7)
  node[midway,below left] {$x=1$};

\node at (3.0,.4) {Local strategies};
\draw[edge] (2.2,0.6)--(2.2,1.05);
\draw[edge] (2.2,0.2)--(2.2,-0.25);
\node at (3.3,.8) {$(a=0,b=0)$};
\node at (3.3,0) {$(a=0,b=y)$};
\node at (5.6,2.8) {{\color{blue}$S\times T=\{0\}\times\{0,1\}$}};
\node at (5.6,-2) {{\color{red}$S\times T=\{1\}\times\{0,1\}$}};

\begin{scope}[shift={(2,1.2)}]
  \draw[line width=.5pt] (0,0) rectangle (2,2);
  \draw[line width=.5pt] (1,0)--(1,2) (0,1)--(2,1);
  \foreach \p in {(.5,1.5),(1.5,1.5),(.5,.5)}
    \node[font=\large] at \p {$\checkmark$};
  \node[font=\large] at (1.5,.5) {$\times$};
  \draw[selected,color=blue] (.07,1.07) rectangle (1.93,1.93);
\end{scope}

\begin{scope}[shift={(2,-2.4)}]
  \draw[line width=.5pt] (0,0) rectangle (2,2);
  \draw[line width=.5pt] (1,0)--(1,2) (0,1)--(2,1);
  \foreach \p in {(.5,1.5),(.5,.5),(1.5,.5)}
    \node[font=\large] at \p {$\checkmark$};
  \node[font=\large] at (1.5,1.5) {$\times$};
  \draw[selected,dashed,color=red] (.07,.07) rectangle (1.93,0.93);
\end{scope}

\draw[edge,color=blue] (4.2,2.5)--(5.,1.5);
\draw[edge,color=red] (4.2,-1.7)--(5.,-0.7);
\begin{scope}[shift={(5.1,-.6)}]
  \draw[line width=.5pt] (0,0) rectangle (2,2);
  \draw[line width=.5pt] (1,0)--(1,2) (0,1)--(2,1);
  \foreach \p in {(.5,1.5),(1.5,1.5),(.5,.5),(1.5,.5)}
    \node[font=\large] at \p {$\checkmark$};
  \draw[selected,color=blue] (.07,1.07) rectangle (1.93,1.93);
  \draw[selected,dashed,color=red] (.07,.07) rectangle (1.93,0.93);
\end{scope}
\node at (6.1,1.65) {One-bit strategy};
\end{tikzpicture}
\caption{A one-bit communication strategy that wins the CHSH game with certainty. Alice and Bob receive inputs $x,y\in\{0,1\}$ and produce outputs $a,b\in\{0,1\}$, respectively. They win when $a\oplus b=xy$. Alice sends her input $x$ to Bob and always outputs $a=0$. Bob outputs $b=0$ if $x=0$ and $b=y$ if $x=1$. The two middle grids show the corresponding local strategies on the full input space. Rows correspond to Alice's inputs $x=0,1$ from top to bottom, and columns to Bob's inputs $y=0,1$ from left to right. Check marks and crosses indicate winning and losing input pairs, respectively. The solid blue and dashed red rectangles mark the {\em restricted} set of input pairs $S\times T$ induced by the communicated bit. Within each $S\times T$, the corresponding local strategy is chosen to satisfy the winning condition for all $(x,y)$, i.e., each $S\times T$ defines a winning rectangle. Together, these winning rectangles partition the input space, giving the one-bit strategy shown on the right.} \label{fig:chsh-winning-rectangles}
\end{figure}
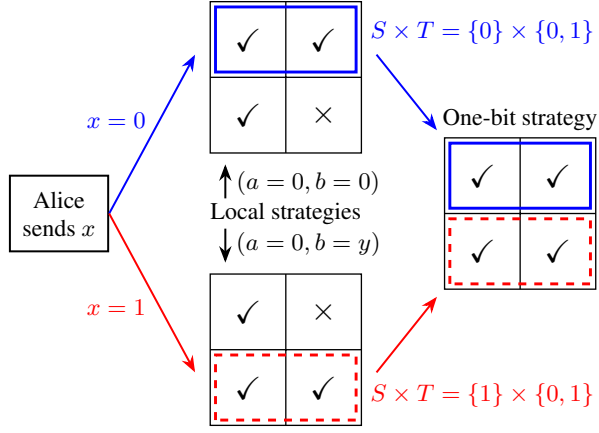

Since each leaf of a perfectly winning protocol corresponds to a winning rectangle, no single leaf can cover more than an $\alpha$ fraction of the total input probability. As the leaves together cover all inputs, the communication must distinguish among sufficiently many leaves. This imposes a lower bound on the average communication, as quantified by the following lemma.

\begin{lemma}[Communication bound from winning rectangles]
\label{lem:single-rectangle-cost-main}
Let $G$ be a game with input set $X\times Y$ and a
full-support input distribution $\mu$. Then
\begin{equation}
 \Cb_\mu^{\leftrightarrow}(G)
 \geq
 \log_2\frac1{\alpha},
 \label{eq:single-rectangle-cost-main}
\end{equation}
where $\alpha$ is the largest probability of a winning
rectangle under $\mu$, as defined in \cref{eq:alpha-main}.
\end{lemma}
The proof is given in Appendix~\ref{app:leaf}, which adapts the standard rectangle-size and entropy argument for average communication~\cite{DietzfelbingerWunderlich2007}.

\subsection{Bounds from winning probabilities}
\label{sec:winning-probability-main}

The winning-rectangle argument applies most directly to correlations that win a game with certainty. To handle more general correlations, we now derive a second bound that depends only on the overall winning probability. Let $\omega_{\mathrm c}:=\omega_{\mathrm c}(G,\mu)\in(0,1)$ denote the largest winning probability achievable without communication.

Our bound is expressed in terms of the binary relative entropy
\begin{equation}
 \Dbin(p\|q)
 :=
 p\log_2\frac pq
 +(1-p)\log_2\frac{1-p}{1-q},
 \label{eq:dbin-definition-app}
\end{equation}
defined for $p\in[0,1]$ and $q\in(0,1)$. For a classical winning
probability $\omega_{\mathrm c}\in(0,1)$, define
\begin{equation}
 \Phi_{\omega_{\mathrm c}}(w)
 :=
 \begin{cases}
 0, & 0\leq w\leq \omega_{\mathrm c},\\[1mm]
 \Dbin(w\|\omega_{\mathrm c}), & \omega_{\mathrm c}<w\leq1.
 \end{cases}
 \label{eq:per-copy-phi-main}
\end{equation}
The communication cost can then be bounded directly from the
winning probability, as shown in the following Lemma.

\begin{lemma}[Communication bound from winning probability]
\label{lem:single-game-cost-main}
Let $G$ have input set $X\times Y$ and a full-support input
distribution $\mu$. Suppose that its classical winning
probability satisfies $\omega_{\mathrm c}\in(0,1)$. Then every
correlation $\vec{P}$ in this Bell scenario satisfies
\begin{equation}
 \Cb_\mu^{\leftrightarrow}(\vec{P})
 \geq
 \Phi_{\omega_{\mathrm c}}\bigl(w_G(\vec{P})\bigr),
 \label{eq:single-game-cost-main}
\end{equation}
where $w_G(\vec{P})$ denotes the winning probability of
$\vec{P}$ in $G$.
\end{lemma}

A winning probability above the classical value therefore implies a positive lower bound on the communication cost. The proof is given in Appendix~\ref{app:single-game-proof}, which adapts standard entropy and relative-entropy arguments used in previous work on information cost and two-prover games~\cite{ChakrabartiKondapallyWang2012,BravermanKo2018}.

The local response functions for a winning rectangle can be extended arbitrarily to all inputs. The resulting zero-communication strategy wins on every input pair in the rectangle and may also win on pairs outside it. Hence, $\alpha\leq\omega_{\mathrm c}$. For a perfectly winning correlation, substituting $w_G(\vec{P})=1$ in Lemma~\ref{lem:single-game-cost-main} therefore yields the lower bound $\log_2(1/\omega_{\mathrm c})$, which cannot exceed the winning-rectangle bound. The winning-probability bound can nevertheless be easier to evaluate when the classical value is known and, unlike the winning-rectangle bound, also applies to correlations that win only imperfectly.

\section{Applications to correlations arising in the parallel repetition of nonlocal games}
\label{sec:parallel-applications-main}

We now use the single-game bounds to study how communication
cost grows when several copies are played in parallel.
In $G^{\otimes n}$, the input pairs are sampled independently
from $\mu$, and each player receives their full input string
and produces one output per copy. Since a joint protocol may use information from all copies together, a communication lower bound cannot generally be
obtained by simply adding the single-game bounds.

\subsection{Winning rectangles in parallel games}
\label{sec:parallel-rectangles-main}

We first consider the winning-rectangle bound under parallel repetition, assuming that every input pair in $X\times Y$ is allowed, and that $\mu$ has full support. Within each copy, Alice's and Bob's inputs may be correlated, while different copies are sampled independently. Let
\begin{equation}
 \alpha_n
 :=
 \max_{S\times T\ {\rm winning}}
 \mu^{\otimes n}(S\times T).
 \label{eq:parallel-alpha-main}
\end{equation}
Here, a winning rectangle $S\times T\subseteq X^n\times Y^n$
admits local response functions, each depending on the
corresponding party's full input string, that win every copy
on every input pair in the rectangle.

The following lemma shows that the largest winning-rectangle
probability is multiplicative across independent blocks
of copies.

\begin{lemma}[Multiplicativity of winning-rectangle probabilities]
\label{lem:tensor-main}
Let $G$ have input set $X\times Y$,
$\mu$ be a full-support distribution on this set, and the input pairs of different copies be independently drawn from $\mu$.
For all positive integers $k$ and $\ell$,
\begin{equation}
 \alpha_{k+\ell}=\alpha_k\alpha_\ell.
\end{equation}
Consequently, for every positive integer $n$,
\begin{equation}
 \alpha_n=\alpha_1^n.
 \label{eq:tensor-main}
\end{equation}
\end{lemma}

The proof is given in Appendix~\ref{app:tensor}.
Applying Lemma~\ref{lem:single-rectangle-cost-main} to
$G^{\otimes n}$ and using~\cref{eq:tensor-main} gives
the following linear lower bound.

\begin{theorem}[Perfect-winning bound for parallel games]
\label{thm:exact-product-main}
Let $G$ have input set $X\times Y$,
$\mu$ be a full-support distribution on this set, and the input pairs of different copies be independently drawn from $\mu$. For every positive integer $n$,
\begin{equation}
\Cb_{\mu^{\otimes n}}^{\leftrightarrow}(G^{\otimes n})
\geq
n\log_2\frac1{\alpha_1},
\label{eq:product-game-main}
\end{equation}
where $\alpha_1$ is the largest probability of a winning rectangle for the single-copy game under $\mu$.
\end{theorem}

\subsection{Winning probabilities in parallel games}
\label{sec:robust-product-main}

We next extend the winning-probability bound to parallel games.
Unlike the winning-rectangle argument, this requires Alice's and
Bob's inputs within each copy to be independent, i.e.,
$\mu=\mu_A\times\mu_B$. The resulting bound controls the total
communication cost in terms of the winning probability of each
copy, even when the protocol processes all copies jointly.

\begin{lemma}[Communication bound for parallel games]
\label{lem:per-copy-cost-main}
Let $G$ have input set $X\times Y$ and full-support product
distribution $\mu=\mu_A\times\mu_B$, with classical winning
probability $\omega_{\mathrm c}\in(0,1)$. For any positive
integer $n$, consider a fully interactive classical protocol
$\mathcal R$ for $G^{\otimes n}$ with inputs drawn from
$\mu^{\otimes n}$. Let $w_i$ denote the probability that
$\mathcal R$ wins copy $i$, averaged over the inputs and $\Lambda$. Then
\begin{equation}
 \E_{\mu^{\otimes n},p_\Lambda}[L_{\mathcal R}]
 \geq
 \sum_{i=1}^n
 \Phi_{\omega_{\mathrm c}}(w_i),
 \label{eq:per-copy-cost-main}
\end{equation}
where $\Phi_{\omega_{\mathrm c}}$ is defined in
\cref{eq:per-copy-phi-main}.
\end{lemma}
The bound depends on the individual winning probabilities and does not require the winning events of different copies to be independent.

Each copy with a winning probability above the classical value contributes positively to the communication lower bound. Moreover, if the protocol wins all $n$ copies simultaneously with probability at least $\omega$, then necessarily $w_i\geq \omega$ for every $i$. Lemma~\ref{lem:per-copy-cost-main} therefore yields the following bound.

\begin{theorem}[Imperfect-winning input-average product bound]
\label{thm:robust-product-main}
Let $G$ have input set $X\times Y$ and a full-support
product distribution $\mu=\mu_A\times\mu_B$, with classical
winning probability $\omega_{\mathrm c}\in(0,1)$.
Then, for every positive integer $n$ and every
$\omega_{\mathrm c}<\omega\leq1$,
\begin{equation}
 \Cb_{\mu^{\otimes n}}^{\leftrightarrow,(\omega)}
 (G^{\otimes n})
 \geq
 n\,\Dbin(\omega\|\omega_{\mathrm c}),
 \label{eq:robust-product-main}
\end{equation}
where $\omega$ denotes the required probability of winning
all $n$ copies simultaneously.
\end{theorem}

The proofs of Lemma~\ref{lem:per-copy-cost-main} and
Theorem~\ref{thm:robust-product-main} are given in
Appendix~\ref{app:robust-product-proof}.

We now apply these two parallel-game bounds to concrete examples, beginning with the perfect-winning case.

\subsection{Magic Square (MS) correlations}
\label{sec:magic-main}

We use the MS game as our principal example of the
perfect-winning regime. In this game, Alice receives a row label
$x\in\{0,1,2\}$ and outputs three bits $a=(a_0,a_1,a_2)$
of even parity, while Bob receives a column label
$y\in\{0,1,2\}$ and outputs three bits $b=(b_0,b_1,b_2)$
of odd parity. They win when their bits agree at the intersection,
\begin{equation}
	 a_y=b_x.
\end{equation}
For local correlations, the associated Bell expression satisfies
\begin{equation}
	 I_{\MS} = \sum_{x,y=0}^{2}P(a_y=b_x|x,y) \leq 8.
\end{equation}
A quantum strategy attains the algebraic maximum of $9$. We denote 
the {\em uniform} input distribution by $\mu_{\MS}$, and the corresponding perfect quantum correlation by $\vec{P}_{\MS}$. The two-ququart strategy realizing $\vec{P}_{\MS}$ is described in Appendix~\ref{app:magic-strategy}.

For classical strategies, perfect winning becomes possible once communication is allowed. It has been shown that a single bit suffices for one copy~\cite{Broadbent06}, while one trit from Alice to Bob suffices for two copies~\cite{Marton24}. This naturally raises the question of how the minimum input-average communication scales with the number of copies when fully interactive protocols are allowed.

To apply the winning-rectangle bound, note that any winning rectangle must omit at least one row or one column, and hence can contain at most six input pairs. This maximum is attained by taking two rows and all three columns, or all three rows and two columns. Assuming $\mu_\MS$, 
the largest single-copy winning-rectangle probability therefore, is
 $2/3$. Lemma~\ref{lem:tensor-main} then gives
\begin{equation}
 \alpha_n=\left(\frac23\right)^n.
 \label{eq:magic-alpha-main}
\end{equation}
Combining this with Theorem~\ref{thm:exact-product-main}
yields one of our main results.

\begin{theorem}[MS simulation lower bound]
\label{thm:finite-budget-main}
For every positive integer $n$, under the uniform input
distribution $\mu_{\MS}^{\otimes n}$, the exact classical
simulation cost of the $n$-copy perfect MS correlation satisfies
\begin{equation}
 \Cb_{\mu_{\MS}^{\otimes n}}^{\leftrightarrow}
 \bigl(\vec{P}_{\MS}^{\otimes n}\bigr)
 \geq
 n\log_2\frac32.
\end{equation}
\end{theorem}

As a direct consequence, Theorem~\ref{thm:finite-budget-main}
provides an explicit threshold for any prescribed finite
input-average communication budget $B\geq0$. In other words, choosing
\begin{equation}
 n_B^{\mathrm{MS}}
 :=
 \left\lfloor
 \frac{B}{\log_2(3/2)}
 \right\rfloor+1
 \label{eq:threshold}
\end{equation}
guarantees that the exact classical simulation of
$\vec{P}_{\MS}^{\otimes n_B^{\mathrm{MS}}}$
requires more than $B$ bits of two-way input-average communication.

To assess the tightness of the lower bound, we construct an
exact one-way simulator of $\vec{P}_{\MS}^{\otimes n}$.
For MS, shared-random output relabelings convert
any perfectly winning protocol into an exact simulator of $\vec{P}_{\MS}$
without additional communication, as shown in
Appendix~\ref{app:output-relabeling}. It therefore suffices
to construct a protocol that wins all $n$ copies with certainty.

For this purpose, we use classical strategies that win whenever one
of Alice's three rows is excluded. For $n$ copies, choose one excluded
row for each of the copies. Each such choice defines a set of Alice inputs
for which the corresponding product strategy wins against every Bob
input. We cover all $3^n$ possible input strings of Alice by sets of
this form (called a \emph{cover}), and let $T(n,3)\le2^n$ denote the minimum number of sets required.

Given such a cover, Alice sends the index of one set compatible with her
input (see Appendix~\ref{app:covering-finite} for the details), after which both parties use the corresponding product strategy.
This gives a perfectly winning one-way protocol with communication
cost $\log_2 T(n,3)$ and, by the relabeling argument above, an exact
simulation of $\vec{P}_{\MS}^{\otimes n}$ with the same cost.

This construction immediately gives the following upper bound on the
exact simulation cost.

\begin{theorem}[MS simulation upper bound]
\label{thm:magic-asymptotic-main}
For every positive integer $n$ and each communication model
$\sigma\in\{\to,\leftrightarrow\}$, under the uniform input
distribution $\mu_{\MS}^{\otimes n}$,
\begin{equation}
 \Cb_{\mu_{\MS}^{\otimes n}}^\sigma
 \bigl(\vec{P}_{\MS}^{\otimes n}\bigr)
 \leq
 \Cwc^\sigma
 \bigl(\vec{P}_{\MS}^{\otimes n}\bigr)
 \leq
 \log_2 T(n,3).
 \label{eq:magic-asymptotic-bracket-main}
\end{equation}
\end{theorem}

To estimate this upper bound, we use the following bounds on the
covering number.

\begin{lemma}[Bounds on the MS covering number]
\label{lem:covering-bound-main}
For every positive integer $n$,
\begin{equation}
 \left\lceil
 \left(\frac32\right)^n
 \right\rceil
 \leq
 T(n,3)
 \leq
 \left\lceil
 (n\ln 3)\left(\frac32\right)^n
 \right\rceil.
 \label{eq:covering-bound-main}
\end{equation}
\end{lemma}

The proof of Lemma~\ref{lem:covering-bound-main} is given in
Appendix~\ref{app:covering-rate}. In particular,
\begin{equation}
 \log_2 T(n,3)
 =
 n\log_2\frac32+O(\log n).
\end{equation}
Combining Theorem~\ref{thm:magic-asymptotic-main} and
Lemma~\ref{lem:covering-bound-main} with the lower bound of
Theorem~\ref{thm:finite-budget-main}, we see that the upper and lower bounds differ
by only $O(\log n)$ bits. Hence, for each
$\sigma\in\{\to,\leftrightarrow\}$,
\begin{equation}
 \lim_{n\to\infty}
 \frac{\Cb_{\mu_{\MS}^{\otimes n}}^\sigma
 \bigl(\vec{P}_{\MS}^{\otimes n}\bigr)}{n}
 =
 \lim_{n\to\infty}
 \frac{\Cwc^\sigma
 \bigl(\vec{P}_{\MS}^{\otimes n}\bigr)}{n}
 =
 \log_2\frac32.
 \label{eq:magic-asymptotic-main}
\end{equation}
Thus, the asymptotic exact-simulation cost is
$\log_2(3/2)$ bits per copy for both input-average and worst-case
communication. In particular, allowing interaction does not reduce
the asymptotic rate, since the same rate is already achieved by the
one-way protocol.

The winning-probability bound provides a corresponding result when the perfect-winning requirement is relaxed. For MS,
$\omega_{\mathrm c}=8/9$ for $\mu=\mu_{\MS}$. Hence, if all $n$ copies are required to
be won simultaneously with probability at least $\omega$, where
$8/9<\omega\leq1$, Theorem~\ref{thm:robust-product-main} gives
\begin{equation}
 \Cb_{\mu_{\MS}^{\otimes n}}^{\leftrightarrow,(\omega)}
 \bigl(\MS^{\otimes n}\bigr)
 \geq
 n\,\Dbin\left(\omega\middle\|\frac89\right).
 \label{eq:robust-magic-average-main}
\end{equation}

For comparison, Jain and Kundu showed that, under a worst-case
communication cap of $C$ bits, the probability of winning more than
an $(8/9+\eta)$ fraction of $n$ parallel MS games is at most
$2^{-\Omega(\eta^3n)+C}$~\cite{JainKundu2025}. For an explicit
numerical comparison, Appendix~\ref{app:rao-jk} evaluates the
finite-$n$ estimate appearing in the proof of Rao's
concentration bound~\cite{Rao2011} and combines it with the same
message-guessing argument, yielding
\begin{equation}
 C_{\mathrm{wc},\mu_{\MS}^{\otimes n}}^{\leftrightarrow,(\omega)}
 (\MS^{\otimes n})
 \geq
 \pos{
 \gamma_{\MS}^{\mathrm R}n
 +
 \log_2\omega-1
 },
\end{equation}
where
$\pos{u}:=\max\{u,0\}$ and
$\gamma_{\MS}^{\mathrm R}\approx2.54\times10^{-8}$. For perfect winning, \(\omega=1\), this explicit lower bound becomes nonzero only when
$n\gtrsim3.937\times10^7$.

In comparison, our winning-probability bound in
\cref{eq:robust-magic-average-main} gives
$n\log_2(9/8)\approx0.169925n$ bits for perfect winning,
while the winning-rectangle bound gives
$n\log_2(3/2)\approx0.584963n$ bits.
Both hold for every positive integer $n$ and are derived
for input-average communication. By
\cref{eq:cost-chain-main}, they also lower-bound the
corresponding worst-case cost.

\subsection{Optimal quantum correlations for the parallel-CHSH game}
\label{sec:chsh-main}

We next turn to the CHSH nonlocal game, where quantum theory achieves an advantage
over classical strategies but does not win with certainty. Specifically, upon receiving the inputs $x,y\in\{0,1\}$ , the players are required to return bits
$a,b\in\{0,1\}$ that satisfy $a\oplus b=xy$. With
uniformly distributed input distributions 
$\mu_{\mathrm{CHSH}}$, the classical winning probability is $3/4$, while the Tsirelson
correlation
\begin{equation}
	 P_T(a,b|x,y) = \frac{1}{4}\left[1+\frac{(-1)^{a\oplus b\oplus xy}}{\sqrt2}\right]
\end{equation}
wins with probability $w_T=\frac{2+\sqrt2}{4}$.

For parallel CHSH, independently applying the optimal single-copy quantum strategy to each copy is already optimal for winning all
copies simultaneously~\cite{Cleve:IEEECCC:2007}. Although this joint winning probability decreases with $n$,
exact simulation of $\vec{P}_T^{\otimes n}$ must still reproduce
the single-copy winning probability $w_T$ in every copy.
By~\cref{eq:per-copy-cost-main}, each such protocol requires at
least $n\Phi_{3/4}(w_T)$ bits of input-average communication. As a result,
\begin{equation}
 \Cb_{\mu_{\mathrm{CHSH}}^{\otimes n}}^{\leftrightarrow}
 \bigl(\vec{P}_T^{\otimes n}\bigr)
 \geq
 n\,\gamma_{\mathrm{CHSH}},
 \label{eq:chsh-per-copy-main}
\end{equation}
for every positive integer $n$, where
\begin{equation}
 \gamma_{\mathrm{CHSH}}
 :=
 \Dbin\left(
 \frac{2+\sqrt2}{4}
 \middle\|
 \frac34
 \right)
 \approx 0.04627385.
 \label{eq:chsh-coefficient-main}
\end{equation}

Thus, parallel CHSH correlations also exceed any prescribed finite
two-way input-average communication budget. For a budget $B\geq0$,
it suffices to choose
\begin{equation}
 n_B^{\mathrm{CHSH}}
 :=
 \left\lfloor
 \frac{B}{\gamma_{\mathrm{CHSH}}}
 \right\rfloor+1,
 \label{eq:chsh-threshold-main}
\end{equation}
in which case the exact simulation of
$\vec{P}_T^{\otimes n_B^{\mathrm{CHSH}}}$
requires more than $B$ bits of two-way input-average communication.

It is useful to compare this with the parallel-CHSH result of M\'arton \emph{et al.}~\cite{Marton24}. They show that exact simulation of $n$ parallel CHSH correlations
requires more than $0.0773n$ bits in a fixed-direction one-way
worst-case communication model. As mentioned in
Ref.~\cite{Marton24}, by symmetry of the CHSH correlation, their
argument also extends to two-way protocols where each round of the game has a fixed communication direction but
successive rounds may use opposite directions. Our bound has the
smaller coefficient
$\gamma_{\mathrm{CHSH}}\approx0.04627$, but directly applies to
variable-length, fully interactive two-way protocols under
input-average communication. Thus, the two bounds address different communication
models and are useful under different assumptions.

\section{Hierarchy of lower bounds}
\label{sec:hierarchy-main}
The preceding bounds use the performance of a correlation in a
chosen nonlocal game. We now turn to a complementary approach that
works directly with the full correlation, following the
convex-decomposition formulation of Ref.~\cite{Pironio:PRA:2003}.
In this formulation, the exact simulation cost is obtained by
decomposing the target correlation into deterministic correlations
and weighting each by its minimum communication cost.

A direct evaluation therefore requires accounting for all compatible
deterministic correlations, whose number grows rapidly with the size
of the Bell scenario~\cite{ZC:Entropy:2019}. To reduce this
computational burden, we construct a linear programming (LP) hierarchy that retains
only deterministic correlations up to a chosen communication cost
and bounds the contribution of the remainder.

To this end, consider a fixed full-support input distribution $\mu$ and a
communication model $\sigma\in\{\to,\leftrightarrow\}$.
Let $\Dset^\sigma$ denote the finite set of deterministic
correlations compatible with this communication model.
For each $\vec{D}\in\Dset^\sigma$, define its minimum
input-average communication cost by minimizing over deterministic
protocols $\Pi^\sigma$ that produce $\vec{D}$:
\begin{equation}
 c_\mu^\sigma(\vec{D})
 :=
 \min_{\Pi^\sigma\to\vec{D}}
 \E_\mu\bigl[L_{\Pi^\sigma}(x,y)\bigr].
 \label{eq:det-cost-main}
\end{equation}

With this notation, the exact input-average cost~\cite{Pironio:PRA:2003} of a target
correlation $\vec{P}$ is obtained by minimizing over all convex
decompositions of $\vec{P}$ into elements of $\Dset^\sigma$,
\begin{equation}
\begin{aligned}
 \Cb_\mu^\sigma(\vec{P}) = \min_{\{q_{\vec{D}}\geq0\}}
 &\quad \sum_{\vec{D}\in\Dset^\sigma} q_{\vec{D}}c_\mu^\sigma(\vec{D}) \\
 \text{such that} &\quad \vec{P}
 = \sum_{\vec{D}\in\Dset^\sigma}q_{\vec{D}}\vec{D}, \\
 &\quad
 \sum_{\vec{D}\in\Dset^\sigma}q_{\vec{D}}=1.
\end{aligned}
\label{eq:exact-hierarchy-main}
\end{equation}

An important feature of this formulation is that deterministic correlations with larger individual communication costs may still
appear with small weight in an optimal decomposition. Consequently, a target correlation can have an input-average communication cost below
a given threshold, say, $C$, even when it cannot be decomposed using only deterministic correlations whose individual costs lie below $C$. Appendix~\ref{app:PR} gives an explicit example of this
phenomenon.

Following Ref.~\cite{Pironio:PRA:2003}, we group deterministic
correlations by their minimum input-average communication costs.
Let's further denote by $|S|$ the cardinality of a set $S$.
Then, as shown in Appendix~\ref{app:finite-simulation-cost}, each cost
is at most $\log_2|X|$ for $\sigma=\to$ and at most
$\log_2|X|+\log_2|Y|$ for $\sigma=\leftrightarrow$.
Since a finite Bell scenario has only finitely many deterministic
correlations, we can order their distinct costs as
\begin{equation}
 0=\kappa_0<\kappa_1<\cdots<\kappa_K.
\end{equation}
We then define
\begin{equation}
 \Dset_k^\sigma
 :=
 \{\vec{D}\in\Dset^\sigma:
 c_\mu^\sigma(\vec{D})=\kappa_k\}.
\end{equation}
The zero-cost class $\Dset_0^\sigma$ is precisely the set
of local deterministic correlations. Moreover, the above observation implies that it suffices to consider, in the computation of~\cref{eq:exact-hierarchy-main}, $\kappa_k=\log_2|X|$ ($\kappa_k=\log_2|X|+\log_2|Y|)$ for $\sigma=\to$ ($\sigma=\leftrightarrow$).

We now define a hierarchy of relaxations indexed by
$\mathfrak{L}=0,\ldots,K-1$. At level $\mathfrak{L}$, we keep only the deterministic
correlations whose costs do not exceed $\kappa_\mathfrak{L}$. Since any omitted
deterministic correlation has cost at least $\kappa_{\mathfrak{L}+1}$, we obtain a lower bound by charging the remaining probability weight at the cost of $\kappa_{\mathfrak{L}+1}$. Relaxing the exact convex decomposition in this way
gives
\begin{equation}
\begin{aligned}
 \Cb_\mu^{\sigma,(\leq\kappa_\mathfrak{L})}(\vec{P})
 :=
 \min_{\{q_{\vec{D}}\geq0\}}
 &\quad
 \sum_{k=0}^{\mathfrak{L}}
 \sum_{\vec{D}\in\Dset_k^\sigma}
 q_{\vec{D}}\kappa_k
 \\
 &\quad+
 \kappa_{\mathfrak{L}+1}
 \left(
 1-
 \sum_{k=0}^{\mathfrak{L}}
 \sum_{\vec{D}\in\Dset_k^\sigma}q_{\vec{D}}
 \right)
 \\
 \text{such that}
 &\quad
 \vec{P}\geq
 \sum_{k=0}^{\mathfrak{L}}
 \sum_{\vec{D}\in\Dset_k^\sigma}
 q_{\vec{D}}\vec{D},
 \\
 &\quad
 \sum_{k=0}^{\mathfrak{L}}
 \sum_{\vec{D}\in\Dset_k^\sigma}q_{\vec{D}}
 \leq1.
\end{aligned}
\label{eq:relaxed-hierarchy-main}
\end{equation}

These relaxations satisfy the ordering stated in the following
theorem, whose proof is given in Appendix~\ref{app:hierarchy}.

\begin{theorem}[Input-average hierarchy]
\label{Thm:LowerBoundHierarchy}
Assume $K\geq1$. For each
$\sigma\in\{\to,\leftrightarrow\}$ and every
$\vec{P}\in\operatorname{conv}(\Dset^\sigma)$,
\begin{equation}
 \Cb_\mu^{\sigma,(\leq\kappa_0)}(\vec{P})
 \leq\cdots\leq
 \Cb_\mu^{\sigma,(\leq\kappa_{K-1})}(\vec{P})
 =
 \Cb_\mu^\sigma(\vec{P}).
 \label{eq:hierarchy}
\end{equation}
\end{theorem}

Thus, the hierarchy reaches the exact input-average communication
cost at its final level. At the other end of the hierarchy, the first level has a simple
interpretation. Any reduced deterministic protocol tree realizing
$\vec{D}\in\Dset^\sigma\setminus\Dset_0^\sigma$ must have a first
message with alphabet size $m\geq2$. This message is sent on every
execution and costs at least one bit, so
$c_\mu^\sigma(\vec{D})\geq1$. Since this bound is saturated by some nonlocal deterministic
correlations, $\kappa_1=1$. The first level therefore minimizes the weight remaining after
subtracting a local contribution. For a no-signaling target, this
remainder is also no-signaling, and the resulting bound equals the
EPR2 nonlocal content~\cite{EPR92} in either communication model.
At higher levels, deterministic correlations with larger minimum
input-average costs are included.

To complement these lower bounds, we obtain upper bounds by
restricting the exact decomposition in
\cref{eq:exact-hierarchy-main} to deterministic correlations
with costs at most $\kappa_k$,
\begin{equation}
\begin{aligned}
 U_{\mu}^{\sigma,(\leq\kappa_k)}(\vec{P})
 :=
 \min_{\{q_{\vec{D}}\geq0\}}
 &\quad
 \sum_{j=0}^{k}
 \sum_{\vec{D}\in\Dset_j^\sigma}q_{\vec{D}}\kappa_j
 \\
 \text{such that}
 &\quad
 \vec{P}=
 \sum_{j=0}^{k}
 \sum_{\vec{D}\in\Dset_j^\sigma}
 q_{\vec{D}}\vec{D}.
\end{aligned}
\label{eq:restricted-upper-main}
\end{equation}
In this case, every feasible decomposition gives a simulation protocol:
the parties use shared randomness to select $\vec{D}$ with
probability $q_{\vec{D}}$ and implement a deterministic protocol
with input-average cost $c_\mu^\sigma(\vec{D})$. The objective
value therefore gives an upper bound on the minimum
input-average communication cost.

If the program is infeasible at threshold $\kappa_k$, it provides
no upper bound at that threshold. Once feasible, increasing the
threshold cannot increase the upper bound. At $\kappa_K$, all
deterministic correlations are included, which recovers the exact
formulation in \cref{eq:exact-hierarchy-main}.
Combining this construction with the lower-bound hierarchy gives,
whenever the upper-bound program is feasible,
\begin{equation}
 \Cb_\mu^{\sigma,(\leq\kappa_{\mathfrak{L}})}(\vec{P})
 \leq
 \Cb_\mu^\sigma(\vec{P})
 \leq
 U_\mu^{\sigma,(\leq\kappa_k)}(\vec{P}).
 \label{eq:cost-bracket-main}
\end{equation}
If the lower bound at some level $\mathfrak L$ coincides with the upper bound at some threshold $k$, they determine the exact input-average communication cost.

As an illustration, we numerically search for quantum correlations
with large first-level lower bounds under uniform inputs $\mu$
in several Bell scenarios. We use MATLAB's \texttt{fminunc}
to optimize over bipartite quantum states and local projective
measurements with local Hilbert-space dimension $n_o$. At each
step, we evaluate $\Cb_\mu^{\sigma,(\leq0)}(\vec{P})$ by solving
the linear program in \cref{eq:relaxed-hierarchy-main} with
$\mathfrak{L}=0$. Table~\ref{Tbl.Summary} presents the largest
values found within this ansatz.

\begin{table}[t]
\centering
\begin{tabular}{c cccc}
\toprule
$n_s$ & \multicolumn{4}{c}{$n_o$} \\
\cmidrule(lr){2-5}
& $2$ & $3$ & $4$ & $5$ \\
\midrule
$2$ & $0.4142$ & $0.4574$ & $0.4882$ & $0.5125$ \\
$3$ & $0.5981$ & $0.6456$ & $1$ & $1^*$ \\
$4$ & $0.6955$ & $0.7467$ & $1^*$  & \\
$5$ & $0.7553$ & ~~$0.8100^*$ & $1^*$  & \\
\bottomrule
\end{tabular}
\caption{Numerical lower bounds found at the first level of the hierarchy
in \cref{eq:hierarchy}, for quantum correlations with
$|X|=|Y|=n_s$ inputs, $|A|=|B|=n_o$ outputs, and local
Hilbert-space dimension $n_o$. For each unstarred entry,
a matching upper bound certifies the two-way input-average
communication cost of the corresponding correlation within
numerical precision. Starred entries remain lower bounds.}
\label{Tbl.Summary}
\end{table}

To assess the tightness of these lower bounds, we solve
\cref{eq:restricted-upper-main} with
$\sigma=\leftrightarrow$ and $\kappa_k=1$, using all zero-
and one-bit deterministic correlations. For each unstarred
entry, the upper bound matches the lower bound within numerical
precision, thereby determining the two-way input-average
communication cost of the corresponding correlation to that
precision. For $(n_s,n_o)=(5,3),(4,4),(5,4)$, and $(3,5)$,
the number of deterministic correlations makes the upper-bound
program exceed our available computational resources. The corresponding values therefore remain lower bounds.

The two-setting values can be compared with the earlier
calculations of Ref.~\cite{Pironio:PRA:2003}. For $n_s=n_o=2$,
we recover the exact CHSH input-average communication cost
$\sqrt{2}-1\simeq0.4142$ obtained there. For
$(n_s,n_o)=(3,4)$, the numerical search recovers the MS
correlation, whose first-level lower bound is exactly one
because its EPR2 nonlocal content is one. A decomposition into
deterministic one-bit correlations gives a matching upper bound.

\section{Discussion}

In this work, we show how properties of quantum correlations can
constrain the communication required for their classical simulation.
The two game-based bounds use complementary properties of the game.
Theorem~\ref{thm:exact-product-main} is based on winning rectangles
for perfect strategies, whereas
Theorem~\ref{thm:robust-product-main} uses winning probabilities in
the more general imperfect-winning case.
The hierarchy, by contrast, works directly with the full distribution $\vec{P}$ 
by retaining deterministic correlations up to a chosen communication
cost and bounding the contribution of the remainder.

Among the examples considered here, MS gives the sharpest
characterization. Earlier parallel-repetition results already imply a linear two-way
classical communication requirement for parallel MS
games~\cite{JainKundu2025,HasegawaLeGallModanese2025}. These results establish the linear scaling but are not aimed at
determining the optimal communication rate. Here, we show that the largest winning-rectangle probability
is multiplicative under parallel repetition, yielding the
finite-copy lower bound $n\log_2(3/2)$ for MS. Together with an
asymptotically matching one-way simulation protocol, this determines
the exact asymptotic simulation rate of $\log_2(3/2)$ bits per copy. Moreover, one-way communication already attains this
rate, so allowing two-way communication does not reduce the
asymptotic cost.

This exact asymptotic rate also provides a natural point of comparison
with existing approaches. In particular, the nonlocal-capacity
framework of Ref.~\cite{MontinaWolf:NJP:2016} gives a general
characterization of the one-way asymptotic exact-simulation cost in
terms of an entropic cost maximized over input distributions. The two approaches use
different cost notions and communication models, and their comparison
helps clarify how the asymptotic simulation cost depends on these
choices.

For finite copy numbers, however, the exact MS-winning simulation costs remain largely open. Nevertheless,
Proposition~\ref{prop:magic-relabeling-app} shows that any perfectly
winning MS protocol can be converted, by output relabeling,
into an exact simulator without additional communication. The
covering values $T(n,3)$ in Appendix~\ref{app:covering-finite} therefore provide finite-copy upper bounds for simulating the full
correlation. In particular, $T(2,3)=3$ recovers the one-trit
winning protocol of Ref.~\cite{Marton24}. Explicit values for up to
nine copies are listed in ~\cref{tab:covering-small-app}.

Beyond the finite-copy cost itself, one can ask how large a Bell
scenario is sufficient to exceed a prescribed communication budget. Under our two-way input-average bounds, a budget of $B$ bits is
exceeded by $n^{\MS}_B\approx1.71B$ copies of MS, compared
with $n^{\mathrm{CHSH}}_B\approx21.61B$ copies of CHSH. In terms of the number of settings, outcomes (local Hilbert space dimension), the smaller number of copies required for MS more than compensates for its slight disadvantage in the single-copy scenario, since $3^{n^{\MS}_B} <4^{n^{\MS}_B}\ll 2^{n^{\mathrm{CHSH}}_B}$. For comparison, the worst-case bound of
M\'arton \emph{et al.}~\cite{Marton24} corresponds to a sufficient
copy number of $\tilde{n}^{\mathrm{CHSH}}_B\approx 12.93B$ for a budget of $B$
bits. Their bound therefore requires fewer CHSH copies than ours, but
is derived for one-way communication with worst-case cost, whereas
our bound allows two-way communication and considers input-average
cost. At large budgets, the MS construction consequently requires fewer
settings and outcomes than either CHSH construction, although all
three scenario sizes still grow exponentially with $B$.

The exponential growth of these constructions naturally raises the
question of whether other families can yield smaller finite
separations. The Deutsch--Jozsa family is a natural candidate, since
its communication lower bounds grow more rapidly with the number of
shared Bell pairs~\cite{Brassard1999,Buhrman98}. Turning these
asymptotic bounds into an explicit scenario size for a prescribed
communication budget, however, would require keeping track of the
corresponding constants. For small instances, perfect local strategies are known for
$m=1,2,3$, whereas zero communication is ruled out for
$m=4$~\cite{Galliard:IEEEISIT:2002,
Galliard:IEEEISIT:2003,Brassard05}. These results establish a finite
separation from local models, but do not by themselves provide an
explicit construction exceeding an arbitrary nonzero input-average
communication budget.

\begin{table}[t]
\centering
\begin{tabular}{c r r r}
\toprule
$n$ & $T(n,3)$
& $n\log_2(3/2)$
& $\log_2T(n,3)$ \\
\midrule
1 & 2  & 0.585 & 1.000 \\
2 & 3  & 1.170 & 1.585 \\
3 & 5  & 1.755 & 2.322 \\
4 & 8  & 2.340 & 3.000 \\
5 & 12 & 2.925 & 3.585 \\
6 & 18 & 3.510 & 4.170 \\
7 & 29 & 4.095 & 4.858 \\
8 & 44 & 4.680 & 5.459 \\
9 & 68 & 5.265 & 6.087 \\
\bottomrule
\end{tabular}
\caption{Winning-rectangle lower bound and two-row-covering upper
bound, in bits, for exact simulation of
$\vec{P}_{\MS}^{\otimes n}$ and for perfectly winning $n$ parallel
MS games. Values are rounded to three decimal places.
The bounds apply to both one-way and two-way communication and to
both input-average and worst-case costs. See Appendix~\ref{app:covering-finite} for details.}
\label{tab:covering-small-app}
\end{table}

A second route toward smaller finite scenarios is to develop the
hierarchy of Theorem~\ref{Thm:LowerBoundHierarchy}, which applies
directly to individual correlations rather than through a chosen
nonlocal game. Its practical strength depends on how tightly the
communication cost can be bounded at computationally accessible
levels. The numerical results in Table~\ref{Tbl.Summary} concern only
the first level and small scenarios. Since this level cannot
certify a communication cost above one bit, higher levels are needed
to obtain larger lower bounds.  

Further improvements may also come from incorporating no-signaling
constraints into communication-assisted Bell expressions and the
relaxations used to evaluate them~\cite{Barrett:PRA:2005,
ZC:Entropy:2019,Tavakoli2023}. Together, these approaches may help
identify quantum correlations whose simulation exceeds a prescribed
communication budget with fewer settings or outcomes.

In addition to reducing the size of the Bell scenario, experimental
applications also require some tolerance to noise. In this respect,
Lemma~\ref{lem:per-copy-cost-main} allows the winning probabilities
to vary across copies. As long as their average exceeds the classical
value by a fixed positive amount, the input-average communication
cost still grows at least linearly with the number of copies. This
remains true even if some copies perform below the classical value,
or if the probability of winning all copies becomes increasingly
small, as can occur for parallel CHSH correlations. It would be
interesting to extend these bounds to include simulation error and
detector loss, building on the connection between communication and
the detection loophole~\cite{Massar:PRA:2002,Laplante:2012}.

A further direction is to extend these results to multipartite
pseudo-telepathy games~\cite{Braverman:SIAM:2025}.  More generally, these results may be useful in settings where strict
space-like separation is replaced by an explicit bound on classical
communication. They constrain classical simulation models with
limited cross-talk or information leakage, including models that
allow fully interactive signaling~\cite{Tabia:2025aa}. Applying such bounds to
device-independent randomness expansion, quantum key distribution, or other device certification tasks~\cite{CCC+24}
would require a more detailed treatment of noise, finite-statistics effects, and the structure of the corresponding protocols. The
methods developed here nevertheless provide a framework for
quantifying how much classical communication is compatible with
observed quantum correlations in finite Bell scenarios.

{\em Note added:} While completing this manuscript, we became aware of the complementary work of~\cite{deGois:arXiv:2026}, which shows that correlations generated by a maximally entangled ququart pair under a continuous family of binary projective measurements cannot be simulated exactly using any fixed finite one-way message alphabet.

\begin{acknowledgments}
We thank Harry Buhrman, Mu-En Liu, Istv\'an M\'arton, K\'aroly F. P\'al, Stefano Pironio, Ashutosh Rai, Peter Sidajaya, and Tam\'as V\'ertesi for helpful discussions.  This work was supported by the National Science and Technology Council, Taiwan (Grants No. 109-2112-M-006-010-MY3, 109-2811-M-006-513, 110-2811-M-006-511, 111-2811-M-006-013, and 112-2628-M-006-007-MY4), the Foxconn Research Institute, Taipei, Taiwan, and in part by the Higher Education Sprout Project, Ministry of Education, to the Headquarters of University Advancement at National Cheng Kung University.
\end{acknowledgments}

\section*{AI use statement}
The authors used ChatGPT-5.6 Sol and ChatGPT-6 Astra to assist with proof ideas, checks of mathematical arguments, and revisions to the writing. These discussions helped improve the communication lower bounds for parallel MS and CHSH correlations and develop an MS upper bound through a connection to a known combinatorial covering problem. All use of these tools was under the authors’ direction. The authors are responsible for the correctness of the results and the final content of the manuscript.

\appendix

\section{Communication costs}
\label{app:costs}

A protocol tree is said to be \emph{reduced} if every internal node
has at least two reachable children, and every branch is reached
by some input. The following lemma shows that any deterministic protocol can be replaced by a finite reduced tree without
increasing the cost of any input pair.

\begin{lemma}[Finite reduced-tree normal form]
\label{lem:finite-tree-hierarchy-app}
For finite $X$ and $Y$, every deterministic protocol
of type $\sigma$ can be reduced without increasing the cost on any
input pair, so that every internal node has at least two reachable
children. Every root-to-leaf path then contains at most $|X|-1$
internal nodes if $\sigma=\to$, and at most $|X|+|Y|-2$ if
$\sigma=\leftrightarrow$. Consequently, for every deterministic
correlation $\vec{D}$ compatible with $\sigma$, its minimum input-average communication cost, cf.~\cref{eq:input-average-cost}, and its minimum
worst-case communication cost, cf.~\cref{eq:worst-case-cost} are {\em finite}.
\end{lemma}

\begin{proof}
In a deterministic protocol tree, each given input pair follows a {\em unique}
sequence of senders and messages, i.e., the sequence defines a transcript. Different input
pairs may produce the same transcript. At any node, the messages exchanged so far separately constrain Alice's and Bob's
possible inputs. Consequently, the input pairs reaching that node
form a rectangle $S\times T$, where $S$ and $T$ are the respective
sets of inputs compatible with those messages. Further discussion
of these input rectangles can be found in
\cref{sec:winning-rectangles-main}.

We now reduce the tree by deleting message symbols that are not
produced by any compatible input and suppressing nodes with only
one reachable child. These operations preserve the outputs and
do not increase the cost on any input pair. At each remaining
Alice's node, the outgoing messages partition $S$ into at least two
nonempty subsets. Along every outgoing edge, Alice's compatible
input set is therefore replaced by a proper subset of $S$, while
$T$ is unchanged. Similarly, at a Bob node, $T$ is replaced by a
proper subset while $S$ is unchanged. Thus, $|S|+|T|$ decreases
by at least one along every edge leaving an internal node.
Starting from $|X|+|Y|$ and ending with nonempty local input sets,
a root-to-leaf path can contain at most $|X|+|Y|-2$ internal nodes.
For a one-way protocol from Alice to Bob, every internal node is
an Alice's node, so $|S|$ decreases by at least one at each step,
giving the sharper bound $|X|-1$.

The number of reachable messages is at most $|X|$ at an
Alice node and at most $|Y|$ at a Bob node. Since both
the depth and message alphabet sizes are bounded, every root-to-leaf path has a finite communication cost.
Moreover, the input and output sets are finite, so only
finitely many reduced deterministic protocol trees need
be considered, up to relabeling of message symbols.
For any fixed deterministic correlation $\vec{D}$
compatible with $\sigma$, its input-average and worst-case
communication costs can therefore be minimized over a
finite set of finite-cost protocols. Both minima are
thus finite.

\end{proof}

By \cref{lem:finite-tree-hierarchy-app}, a finite Bell scenario admits
only finitely many reduced deterministic protocol types. For any $\Lambda=\lambda$, the protocol reduces to one of these
deterministic types. $\lambda$ that produce the same reduced
protocol can therefore be combined by adding their probabilities.
Hence, it is sufficient to consider $\Lambda$ with finite
support.

Recall that $\Cb_\mu^\sigma(\vec{P})$ and
$\Cwc^\sigma(\vec{P})$ denote the input-average and worst-case
communication costs defined in \cref{sec:preliminaries}. An intermediate
notion commonly used in communication complexity is the worst-input
expected cost,
\begin{equation}
 \Cexp^{\sigma}(\vec{P})
 :=
 \min_{\mathcal R^{\sigma}\to\vec{P}}
 \max_{x,y}
 \E_{p_\Lambda}
 [L_{\mathcal R^\sigma}(x,y,\Lambda)].
\end{equation}
The three quantities satisfy
\begin{equation}
 \Cb_\mu^{\sigma}(\vec{P})
 \leq
 \Cexp^{\sigma}(\vec{P})
 \leq
 \Cwc^{\sigma}(\vec{P}).
\end{equation}
The distinction lies in the order of averaging and maximization.
The input-average cost averages over both the inputs and $\Lambda$, whereas $\Cexp^\sigma$ averages over $\Lambda$ before maximizing over the inputs. The worst-case cost
bounds every execution. 

All three cost notions use the same convention for assigning a cost to
each message. A message with alphabet size $m$ is assigned a cost of
$\log_2 m$ bits, so binary protocols correspond to $m=2$ at every
internal node. When the message probabilities are nonuniform, a
variable-depth binary encoding may have a smaller average cost than a
fixed-length encoding. Variable-depth encodings are therefore included in our protocol model, with the communication cost determined by the length of the
path actually taken.

\subsection{Message alphabets and finite simulation cost}
\label{app:finite-simulation-cost}

For a deterministic correlation $\vec{D}\in\Dset^\to$,
Alice's output depends only on $x$, while Bob's output
may depend on both $x$ and $y$. Alice can send her input
$x$ using an alphabet of size $|X|$, partitioning the
input space into rectangles $\{x\}\times Y$.
She can determine her output from her own input, while
Bob knows both inputs after receiving $x$ and can therefore
determine his output as well. Thus, no further communication
is needed, giving $c_\mu^\to(\vec{D})\leq\log_2|X|$.

For $\vec{D}\in\Dset^{\leftrightarrow}$, Alice and Bob
can exchange their inputs, partitioning the input space
into rectangles $\{x\}\times\{y\}$, each containing only
one input pair. Both parties then know the full input pair
and can produce the outputs specified by $\vec{D}$
without further communication. Hence,
$c_\mu^{\leftrightarrow}(\vec{D})
\leq\log_2|X|+\log_2|Y|$.

Both bounds hold for any input distribution $\mu$,
since the constructed protocols have the same communication
cost on every input pair.

\section{Proof of the winning-rectangle bound}
\label{app:leaf}

This appendix proves Lemma~\ref{lem:single-rectangle-cost-main}.
We use two standard properties of deterministic communication
protocols~\cite{KushilevitzNisan1997}. First, for any $\Lambda = \lambda$, the set of input pairs consistent with any transcript is a
rectangle $S\times T$. Second, every finite reduced protocol tree
satisfies
\begin{equation}
 \sum_{\ell\in\mathrm{Leaves}}2^{-L(\ell)}=1.
 \label{eq:kraft-app}
\end{equation}
Indeed, assign weight one to the root. At each internal node $v$,
divide its weight equally among its $m(v)$ children, so each child
receives a fraction $2^{-\log_2 m(v)}$ of its parent's weight.
Multiplying these factors along each root-to-leaf path gives
the leaf weight $2^{-L(\ell)}$. Since each branching preserves
the total weight, the leaf weights sum to one, giving
\cref{eq:kraft-app}.

\begin{proof}[Proof of Lemma~\ref{lem:single-rectangle-cost-main}]
For any $\Lambda=\lambda$, we may, without loss of generality, take the protocol to be deterministic, then the inputs reaching each leaf $\ell$ form
a rectangle. Since the protocol wins with certainty and $\mu$ has
full support, this rectangle must be a winning rectangle. Hence, if $p_\ell :=  \Prob_{\mu}[\text{reach }\ell]$,
then $p_\ell\leq\alpha$ for every leaf.

We next relate the entropy of the leaf distribution to the
communication cost. Define $q_\ell:=2^{-L(\ell)}$.
By \cref{eq:kraft-app}, we have $ \sum_\ell q_\ell=1$.
The following is the standard entropy argument for average
communication~\cite{DietzfelbingerWunderlich2007}. The nonnegativity
of the relative entropy of $p=\{p_\ell\}_\ell$  from $q=\{q_\ell\}_\ell$ gives
\begin{align}
	 0 \leq D(p\Vert q) 
	 &= \sum_{\ell} p_\ell\log_2\frac{p_\ell}{q_\ell} \nonumber\\
	 &= -H(p)+\sum_\ell p_\ell L(\ell),
\end{align}
and hence $\E_\mu[L|\Lambda=\lambda] \geq H(p)$.
Moreover, since $p_\ell\leq\alpha$ for every leaf,
\begin{equation}
	 H(p) = \sum_{\ell} p_\ell\log_2\frac1{p_\ell} \geq  \log_2\frac1{\alpha}.
\end{equation}
Thus,
\begin{equation}
	 \E_\mu[L|\Lambda=\lambda] \geq \log_2\frac1{\alpha}.
 \label{eq:rectangle-entropy-app}
\end{equation}

For a randomized protocol that wins with certainty, every value
$\lambda$ occurring with positive probability induces a deterministic protocol that also wins with certainty. Applying
\cref{eq:rectangle-entropy-app} for each such $\lambda$ and averaging over the protocol randomness gives
\begin{equation}
	 \E_{\mu,p_\Lambda}[L] \geq \log_2\frac1{\alpha}.
\end{equation}
Minimizing over all perfectly winning protocols proves the claim.
\end{proof}

\section{Bounds from winning probabilities}
\label{app:robust-average}

\subsection{Communication bound for a single game}
\label{app:single-game-proof}

\begin{proof}[Proof of Lemma~\ref{lem:single-game-cost-main}]
Consider any classical protocol for $G$. Let $Z=(X,Y)$ denote the
input pair, distributed according to $\mu$, and let $\Lambda$ collect
all protocol randomness, distributed according to $p_\Lambda$ and
independent of the inputs. Let $T$ be the random variable representing a complete transcript from the root to a leaf and
$L$ the communication cost.

For each fixed $\Lambda=\lambda$, the protocol is deterministic.
The entropy argument used in the proof of
Lemma~\ref{lem:single-rectangle-cost-main} therefore gives
\begin{equation}
 \E_\mu[L|\Lambda=\lambda]
 \geq
 H(T|\Lambda=\lambda),
\end{equation}
where we have used the fact that a complete transcript  $T$ is in one-to-one correspondence with some path $\ell$. Averaging over $\lambda$ while using the relation between conditional mutual information and conditional entropies yields
\begin{equation}
	 \E_{\mu,p_\Lambda}[L] \geq H(T|\Lambda) = I(Z;T|\Lambda).
 \label{eq:single-game-information-app}
\end{equation}
Here the equality follows from $H(T|Z,\Lambda)=0$, since the inputs
and protocol randomness completely determine the transcript.

Since $\Lambda$ is independent of $Z$, the mutual information admits
the representation
\begin{equation}
 I(Z;T|\Lambda)
 =
 \E_{p_{T,\Lambda}}
 \left[
 D\left(
 P_{Z|T,\Lambda}\middle\|\mu
 \right)
 \right].
 \label{eq:single-game-posterior-information-app}
\end{equation}
Here $p_{T,\Lambda}$ denotes the joint distribution of the transcript
and protocol randomness induced by the protocol. Thus,
$I(Z;T|\Lambda)$ is the average relative entropy between the
conditional input distribution $P_{Z|T,\Lambda}$ and the original
input distribution $\mu$.

We next relate a change in the input distribution to the resulting
change in winning probability. Consider a fixed local strategy, and
let $f(x,y)$ denote its probability of winning on input $(x,y)$. Under the original input distribution
$\mu$, its winning probability is
\begin{equation}
 r
 :=
 \sum_{x,y}\mu(x,y)f(x,y)
 \leq
 \omega_{\mathrm c}.
\end{equation}
Under another input distribution $\nu$, the same strategy wins with
probability
\begin{equation}
 w
 :=
 \sum_{x,y}\nu(x,y)f(x,y).
\end{equation}
Applying the same binary channel to both input distributions, with
``win'' occurring with probability $f(x,y)$ and ``loss'' otherwise,
the data-processing inequality gives
\begin{equation}
 D(\nu\Vert\mu)
 \geq
 \Dbin(w\Vert r).
\end{equation}

Since our goal is to obtain a communication lower bound from the winning
probability alone, a winning probability not exceeding
$\omega_{\mathrm c}$ does not certify any need for communication, as it
does not surpass the classical value. We therefore take the lower bound
to be zero when $w\leq\omega_{\mathrm c}$.

For $w>\omega_{\mathrm c}$, we have
$r\leq\omega_{\mathrm c}<w$. For fixed $w$, the binary relative
entropy decreases as its second argument approaches $w$, and hence
\begin{equation}
 \Dbin(w\Vert r)
 \geq
 \Dbin(w\Vert\omega_{\mathrm c}).
\end{equation}
Using $\Phi_{\omega_{\mathrm c}}$ defined in
\cref{eq:per-copy-phi-main}, the two regimes can therefore be summarized as
\begin{equation}
 D(\nu\Vert\mu)
 \geq
 \Phi_{\omega_{\mathrm c}}(w).
 \label{eq:local-winning-divergence-app}
\end{equation}
This inequality holds for any fixed local strategy and any input
distribution $\nu$.

We now apply this comparison transcript by transcript. Fix $t$ and
$\lambda$ such that
$\Prob[T=t|\Lambda=\lambda]>0$. The input pairs consistent with
transcript $t$ form a rectangle
$A_{t,\lambda}\times B_{t,\lambda}$. On this rectangle, the outputs
are determined by local functions
$a_{t,\lambda}(x)$ and $b_{t,\lambda}(y)$. Extend these functions
arbitrarily to all local inputs. For fixed $t$ and $\lambda$, the
extended functions therefore define a local strategy without
communication on the full input space.

For this local strategy, let $f_{t,\lambda}(x,y)$ denote its
probability of winning on input $(x,y)$. Conditioned on the transcript
$t$ and randomness $\lambda$, the input pair is distributed according to
\begin{equation}
 \nu_{t,\lambda}(x,y)
 :=
 \Prob[Z=(x,y)|T=t,\Lambda=\lambda],
 \label{eq:update-distribution}
\end{equation}
which may vanish for some input pairs. Its winning probability under this conditional distribution is
\begin{equation}
 w_{t,\lambda}
 :=
 \sum_{x,y}
 \nu_{t,\lambda}(x,y)f_{t,\lambda}(x,y)
 =
 \Prob[\text{win }G|T=t,\Lambda=\lambda].
\end{equation}

Under the original input distribution $\mu$, the same local strategy
wins with probability
\begin{equation}
 r_{t,\lambda}
 :=
 \sum_{x,y}\mu(x,y)f_{t,\lambda}(x,y)
 \leq
 \omega_{\mathrm c}.
\end{equation}
Applying \cref{eq:local-winning-divergence-app} with
$\nu=\nu_{t,\lambda}$ therefore gives
\begin{equation}
 D(\nu_{t,\lambda}\Vert\mu)
 \geq
 \Phi_{\omega_{\mathrm c}}(w_{t,\lambda}).
 \label{eq:single-game-posterior-app}
\end{equation}

We now average the bound in
\cref{eq:single-game-posterior-app} over the transcript and protocol
randomness. The function $\Phi_{\omega_{\mathrm c}}$ is convex and nondecreasing.
It is constant on $[0,\omega_{\mathrm c}]$, while
$\Dbin(w\Vert\omega_{\mathrm c})$ is convex and increasing for
$w>\omega_{\mathrm c}$, with the two branches joining smoothly at
$w=\omega_{\mathrm c}$. Using
\cref{eq:single-game-information-app,eq:single-game-posterior-app,eq:update-distribution,eq:single-game-posterior-information-app}
and the convexity of $\Phi_{\omega_{\mathrm c}}$, we obtain
\begin{align}
 \E_{\mu,p_\Lambda}[L]
 &\geq I(Z;T|\Lambda)
 \nonumber\\
 &\geq
 \E_{p_{T,\Lambda}}
 [\Phi_{\omega_{\mathrm c}}(w_{T,\Lambda})]
 \nonumber\\
 &\geq
 \Phi_{\omega_{\mathrm c}}
 \left(
 \E_{p_{T,\Lambda}}w_{T,\Lambda}
 \right)
 \nonumber\\
 &=
 \Phi_{\omega_{\mathrm c}}
 \bigl(\Prob[\text{win }G]\bigr).
 \label{eq:single-game-protocol-app}
\end{align}
If the protocol exactly simulates $\vec{P}$, then
$\Prob[\text{win }G]=w_G(\vec{P})$. Since the resulting lower bound
holds for every exact simulation protocol, it also bounds the
minimum communication cost, proving \cref{eq:single-game-cost-main}.
\end{proof}

\subsection{Extension to parallel games}
\label{app:robust-product-proof}

We next extend the winning-probability bound to parallel games.
Unlike the single-game result, the argument below requires Alice's
and Bob's inputs within each copy to be independent, so we assume
$\mu=\mu_A\times\mu_B$.

\begin{proof}[Proof of Lemma~\ref{lem:per-copy-cost-main}]
Let $Z_i=(X_i,Y_i)$ and $Z^n=(Z_1,\ldots,Z_n)$, where the
$Z_i$ are independent and distributed according to
$\mu=\mu_A\times\mu_B$. Let $\Lambda$ collect all protocol
randomness, and let $T$ and $L$ denote the complete transcript and
communication cost.

As in the first inequality of \cref{eq:single-game-protocol-app},
\begin{equation}
 \E_{\mu^{\otimes n},p_\Lambda}[L]
 \geq
 I(Z^n;T|\Lambda).
 \label{eq:info-step1}
\end{equation}
Since the copies are independent, subadditivity of conditional entropy
then gives
\begin{align}
 I(Z^n;T|\Lambda)
 &=
 \sum_{i=1}^n H(Z_i)-H(Z^n|T,\Lambda)
 \nonumber\\
 &\geq
 \sum_{i=1}^n I(Z_i;T|\Lambda).
 \label{eq:info-step2}
\end{align}

Fix $i$, $t$, and $\lambda$ with
$\Prob[T=t|\Lambda=\lambda]>0$. The full input strings consistent
with $t$ form a rectangle
$A_{t,\lambda}\times B_{t,\lambda}$. From the output
functions at this leaf, construct a one-copy local strategy as follows.
Given $x_i$, Alice samples her remaining inputs according to their
conditional distribution given $X_i=x_i$, $T=t$, and
$\Lambda=\lambda$, and Bob does the same independently given
$Y_i=y_i$, $T=t$, and $\Lambda=\lambda$. Each party then applies
the output function and keeps the $i$th output.

Because $\mu=\mu_A\times\mu_B$ and the transcript defines a rectangle,
these independent local completions reproduce the conditional
distribution of the remaining inputs. Hence, with
\begin{equation}
 \nu_{i,t,\lambda}(x,y)
 :=
 \Prob[Z_i=(x,y)|T=t,\Lambda=\lambda],
\end{equation}
the winning probability of this local strategy under
$\nu_{i,t,\lambda}$ is
\begin{equation}
 w_{i,t,\lambda}
 =
 \Prob[W_i|T=t,\Lambda=\lambda],
\end{equation}
where $W_i$ denotes the event that copy $i$ is won. Under the original distribution $\mu$, the same local strategy has
winning probability at most $\omega_{\mathrm c}$. Therefore, by
\cref{eq:local-winning-divergence-app},
\begin{equation}
 D(\nu_{i,t,\lambda}\Vert\mu)
 \geq
 \Phi_{\omega_{\mathrm c}}(w_{i,t,\lambda}).
\end{equation}

Applying the same averaging and convexity argument as in the
single-game proof gives
\begin{equation}
 I(Z_i;T|\Lambda)
 \geq
 \Phi_{\omega_{\mathrm c}}
 \bigl(\Prob[W_i]\bigr).
 \label{eq:per-copy-info-app}
\end{equation}
Combining
\cref{eq:info-step1,eq:info-step2,eq:per-copy-info-app} yields
\begin{equation}
 \E_{\mu^{\otimes n},p_\Lambda}[L]
 \geq
 \sum_{i=1}^n
 \Phi_{\omega_{\mathrm c}}
 \bigl(\Prob[W_i]\bigr),
\end{equation}
which proves Lemma~\ref{lem:per-copy-cost-main}.
\end{proof}

\begin{proof}[Proof of Theorem~\ref{thm:robust-product-main}]
Consider any protocol that wins all $n$ copies simultaneously with
probability at least $\omega$. Its probability $w_i$ of winning copy
$i$ satisfies $w_i\geq\omega$ for every $i$. Since
$\omega>\omega_{\mathrm c}$ and
$\Phi_{\omega_{\mathrm c}}$ is nondecreasing,
Lemma~\ref{lem:per-copy-cost-main} gives
\begin{align}
\E_{\mu^{\otimes n},p_\Lambda}[L_{\mathcal R}]
&\geq \sum_{i=1}^n \Phi_{\omega_{\mathrm c}}(w_i)
\nonumber\\
&\geq \sum_{i=1}^n \Phi_{\omega_{\mathrm c}}(\omega)
\nonumber\\
&= n\Dbin(\omega\|\omega_{\mathrm c}).
\end{align}
Since the resulting lower bound
holds for every protocol giving $\omega$, it also bounds the
minimum communication cost, proving \cref{eq:robust-product-main}.
\end{proof}

\section{Multiplicativity of winning-rectangle probabilities}
\label{app:tensor}

\begin{proof}[Proof of Lemma~\ref{lem:tensor-main}]
For the lower bound, take optimal winning rectangles
$S_k\times T_k$ and $S_\ell\times T_\ell$.
Their product gives the explicit winning rectangle
$(S_k\times S_\ell)\times(T_k\times T_\ell)$ for $k+\ell$ copies.
Its probability is $\alpha_k\alpha_\ell$, and hence
\begin{equation}
 \alpha_{k+\ell}\geq\alpha_k\alpha_\ell.
\end{equation}

For the reverse inequality, let $S\times T$ be an optimal winning
rectangle for $k+\ell$ copies, so that
\begin{equation}
 \mu^{\otimes(k+\ell)}(S\times T)
 =
 \alpha_{k+\ell}.
\end{equation}
Write Alice's input as
$(u,u')\in X^k\times X^\ell$ and Bob's as
$(v,v')\in Y^k\times Y^\ell$.
Let $U$ and $V$ denote the projections of $S$ and $T$ onto their first
$k$ coordinates. For each $u\in U$ and $v\in V$, choose local
extensions $u'(u)$ and $v'(v)$. Applying the original response
functions to these extended inputs and retaining only the first $k$
outputs shows that $U\times V$ is a winning rectangle for the first
$k$ copies. Hence,
\begin{equation}
 \mu^{\otimes k}(U\times V)
 \leq
 \alpha_k.
 \label{eq:projection-app}
\end{equation}

For $u\in U$ and $v\in V$, define
\begin{align}
 S_u
 &:=
 \{u'\in X^\ell:(u,u')\in S\},
 \nonumber\\
 T_v
 &:=
 \{v'\in Y^\ell:(v,v')\in T\}.
\end{align}
Fixing $(u,v)$ in the original response functions and retaining only
the last $\ell$ outputs shows that $S_u\times T_v$ is a winning
rectangle for the last $\ell$ copies. Therefore,
\begin{equation}
 \mu^{\otimes\ell}(S_u\times T_v)
 \leq
 \alpha_\ell.
 \label{eq:fiber-app}
\end{equation}

Using independence between the first $k$ and last $\ell$ copies,
\begin{align}
 \alpha_{k+\ell}
 &=
 \mu^{\otimes(k+\ell)}(S\times T)
 \nonumber\\
 &=
 \sum_{u\in U}\sum_{v\in V}
 \mu^{\otimes k}(u,v)
 \mu^{\otimes\ell}(S_u\times T_v)
 \nonumber\\
 &\leq
 \alpha_\ell\,
 \mu^{\otimes k}(U\times V)
 \nonumber\\
 &\leq
 \alpha_k\alpha_\ell.
\end{align}
Combining this with the lower bound gives
\begin{equation}
 \alpha_{k+\ell}
 =
 \alpha_k\alpha_\ell.
\end{equation}
Iterating the identity gives
$\alpha_n=\alpha_1^n$.
\end{proof}

\section{MS correlations and communication cost}
\label{app:magic}

\subsection{Quantum strategy}
\label{app:magic-strategy}

We use the MS game and notation introduced in \cref{sec:magic-main}, with uniform and independent inputs.
A perfect quantum strategy uses
\begin{equation}
\ket{\Phi_4}
=
\frac12\sum_{j=0}^{3}\ket j_A\ket j_B
\end{equation}
and the Mermin--Peres array~\cite{Peres1990,Mermin_90}
\begin{equation}
\begin{array}{ccc}
-Z\otimes I & -I\otimes Z & Z\otimes Z\\
I\otimes X & X\otimes I & X\otimes X\\
Z\otimes X & X\otimes Z & Y\otimes Y
\end{array}.
\label{eq:magic-observables-app}
\end{equation}
The three observables in each row commute and have product $+I$,
whereas those in each column commute and have product $-I$.
Alice jointly measures the observables in her row, while Bob jointly
measures the transposes of the observables in his column. Their three
eigenvalues are converted to output bits by mapping $+1$ to $0$ and
$-1$ to $1$. The row and column product constraints therefore give
Alice even parity and Bob odd parity.

For every observable $M$ in the array, $M^2=I$. Together with
\begin{equation}
 (M\otimes I)\ket{\Phi_4}
 =
 (I\otimes M^{\mathsf T})\ket{\Phi_4},
\end{equation}
this implies
\begin{equation}
 (M\otimes M^{\mathsf T})\ket{\Phi_4}
 =
 \ket{\Phi_4}.
\end{equation}
Hence Alice's and Bob's eigenvalues for the observable at the
intersection are identical with certainty. After mapping the
eigenvalues to bits, their intersection outputs therefore agree.

To obtain the full output distribution, label the rows and
columns of \cref{eq:magic-observables-app} by $0,1,2$ and
denote the observable in row $x$ and column $y$ by $O_{xy}$.
For an even-parity output $a=(a_0,a_1,a_2)$, Alice's joint
measurement projector is~\cite{Kundu:Quantum:2022}
\begin{equation}
 A_{a|x}
 =
 \frac14\left(I+\sum_{j=0}^{2}(-1)^{a_j}O_{xj}\right).
\end{equation}
Similarly, for an odd-parity output $b=(b_0,b_1,b_2)$,
Bob's joint measurement projector is
\begin{equation}
 B_{b|y}
 =
 \frac14\left(I+\sum_{i=0}^{2}(-1)^{b_i}
 O_{iy}^{\mathsf T}\right).
\end{equation}
These expressions follow from the commutation and product
constraints within each row and column.

The nine observables are traceless and mutually orthogonal
under the trace inner product,
\begin{equation}
 \operatorname{tr}(O_{xj}O_{iy})
 =
 4\delta_{xi}\delta_{jy}.
\end{equation}
Using
$\langle\Phi_4|A\otimes B|\Phi_4\rangle
=\frac14\operatorname{tr}(AB^{\mathsf T})$,
we therefore obtain
\begin{align}
 P_{\MS}(a,b|x,y)
 &=
 \langle\Phi_4|A_{a|x}\otimes B_{b|y}|\Phi_4\rangle
 \nonumber\\
 &=
 \frac14\operatorname{tr}
 \left(A_{a|x}B_{b|y}^{\mathsf T}\right)
 \nonumber\\
 &=
 \frac{1+(-1)^{a_y+b_x}}{16}.
 \label{eq:magic-output-probability-app}
\end{align}
Thus, for each input pair, the eight winning output pairs
each occur with probability $1/8$, while all other output
pairs have probability zero.

\subsection{Output relabeling}
\label{app:output-relabeling}

\begin{proposition}[Output relabeling]
\label{prop:magic-relabeling-app}
For every $n\geq1$, any protocol that wins all $n$ MS games
with certainty for every input can be converted into an exact simulator
of $\vec{P}_{\MS}^{\otimes n}$ using only shared randomness. No
additional communication is required, and both the communication
model and the cost of every execution are preserved.
\end{proposition}

\begin{proof}
We first consider a single copy. As shown in
\cref{app:magic-strategy}, for each input pair $(x,y)$, the quantum
strategy produces the correlation
\cref{eq:magic-output-probability-app}, which is uniform over the
eight winning output pairs. It therefore suffices to show that any
winning output pair on inputs $(x,y)$ can be mapped, using shared randomness alone,
to a uniformly random winning pair.

Alice and Bob use shared randomness, independent of the inputs and of
the original protocol, to choose \emph{uniformly} a binary $3\times3$ matrix
$H$ whose rows and columns all have even parity. There are $16$ such
matrices: the four upper-left entries may be chosen freely, after which
the parity conditions uniquely determine the remaining entries.

Suppose the original protocol produces a winning pair $(a,b)$ on
inputs $(x,y)$. Alice adds row $x$ of $H$ to her output, while Bob adds
column $y$ of $H$ to his output, with addition modulo two:
\begin{equation}
 a'_j=a_j\oplus H_{xj},
 \qquad
 b'_i=b_i\oplus H_{iy}.
\end{equation}
Because every row and column of $H$ has even parity, the required
parities of Alice's and Bob's outputs are preserved. Moreover, at the
intersection both parties add the same bit $H_{xy}$, so
$a'_y=b'_x$ whenever $a_y=b_x$. Hence the relabeled outputs still form
a winning pair.

It remains to show that the eight winning pairs occur with equal
probability. Fix an original winning pair $(a,b)$ and a target winning
pair $(a',b')$. To map $(a,b)$ to $(a',b')$, the entries of $H$ along
row $x$ and column $y$ must satisfy
\begin{equation}
 H_{xj}=a_j\oplus a'_j,
 \qquad
 H_{iy}=b_i\oplus b'_i.
\end{equation}
Thus, the entire row $x$ and column $y$ are fixed. Since both the
original and target outputs satisfy the required parities, the
prescribed row and column have even parity. Their prescriptions also
agree at the intersection, because
\begin{equation}
 a_y=b_x,
 \qquad
 a'_y=b'_x.
\end{equation}

It remains only to complete the other four entries of $H$. Relabeling
the rows and columns so that the prescribed row and column come first,
write the fixed entries as
\begin{equation}
 H=
 \begin{pmatrix}
 p & q & p\oplus q\\
 r & * & *\\
 p\oplus r & * & *
 \end{pmatrix}.
\end{equation}
Once the upper-left entry of the remaining $2\times2$ block is chosen,
say $u\in\{0,1\}$, all other entries are uniquely determined by the
even-parity conditions:
\begin{equation}
 H=
 \begin{pmatrix}
 p & q & p\oplus q\\
 r & u & r\oplus u\\
 p\oplus r & q\oplus u & p\oplus q\oplus r\oplus u
 \end{pmatrix}.
\end{equation}
There are therefore exactly two compatible matrices, corresponding to
$u=0$ and $u=1$. Since $H$ is chosen uniformly from the $16$
even-parity matrices, any prescribed winning pair $(a',b')$ is obtained
with probability $ \frac{2}{16}=\frac18$.
Hence, the relabeled output is uniform over the eight winning pairs,
independently of the original winning output, as required by
\cref{eq:magic-output-probability-app}.

For $n$ copies, apply the same relabeling independently to each copy,
using independent matrices $H_1,\ldots,H_n$. Conditional on any fixed
input and any winning output of the original protocol, each copy is
mapped uniformly to one of its eight winning pairs. Independence of
the matrices therefore gives probability $(1/8)^n$ for every
specified winning output across all $n$ copies. This is precisely the
product correlation $\vec{P}_{\MS}^{\otimes n}$.

Since this distribution is independent of the original winning output,
averaging over the original protocol still yields
$\vec{P}_{\MS}^{\otimes n}$. The relabeling changes only the final
outputs, so all communicated messages, and hence the communication
model and cost of every execution, are unchanged.
\end{proof}

\subsection{One-way simulation from two-row coverings}
\label{app:covering-finite}

We now give the details of the covering construction used in the main
text. By \cref{prop:magic-relabeling-app}, it is sufficient to construct
a one-way protocol that wins all $n$ MS games with certainty.

Label the three rows and columns by $0,1,2$. For a single copy,
restricting Alice to any two rows removes the parity contradiction.
Thus, for each excluded row $r\in\{0,1,2\}$, there exists a
deterministic strategy that wins whenever Alice's input satisfies
$x\neq r$.

For $n$ copies, choose an excluded row for each copy and denote the excluded row for the $i$-th copy by $r_i$:
\begin{equation}
 r=(r_1,\ldots,r_n)\in\{0,1,2\}^n.
\end{equation}
This pattern can be used for Alice's input
$x=(x_1,\ldots,x_n)$ whenever
\begin{equation}
 x_j\neq r_j
 \qquad
 \text{for every }j.
\end{equation}
Accordingly, define the corresponding two-row box
\begin{equation}
 B_r
 :=
 \{x\in\{0,1,2\}^n:x_j\neq r_j\ \forall j\}.
\end{equation}
Each box contains $2^n$ inputs.

To cover all possible inputs of Alice, define
\begin{align}
 \tau_n
 :=
 \min\biggl\{
 |R|:
 &R\subseteq\{0,1,2\}^n,\\
 &\bigcup_{r\in R}B_r=\{0,1,2\}^n
 \biggr\}.
 \label{eq:tau-app}
\end{align}
Fix a minimum cover $R$. On input $x$, Alice chooses the first
$r\in R$ such that $x\in B_r$, sends its index, and the parties apply
the corresponding excluded-row strategy independently in each
coordinate. The protocol therefore wins on every input pair using a
single message from an alphabet of size $\tau_n$. Hence,
\begin{equation}
 \Cb_{\mu_{\MS}^{\otimes n}}^{\to}(\MS^{\otimes n})
 \leq
 C_{\mathrm{wc},\mu_{\MS}^{\otimes n}}^{\to}(\MS^{\otimes n})
 \leq
 \log_2\tau_n.
 \label{eq:tau-cost-app}
\end{equation}

The quantity $\tau_n$ is a known covering number with several
equivalent descriptions. Consider the ternary Hamming space $\{0,1,2\}^n$, with
distance $d_{\mathrm H}(x,r)$ defined as the number of coordinates
at which $x$ and $r$ differ~\cite{Brink2011}. Each set $B_r$ contains exactly the strings
at maximum Hamming distance from $r$,
\begin{equation}
 B_r
 =
 \{x:d_{\mathrm H}(x,r)=n\}.
\end{equation}
Thus, $\tau_n$ is the inverse football-pool number
$T(n,3)$~\cite{OstergardRiihonen2003,Brink2011}, i.e., the minimum
number of bets such that every outcome vector is completely missed by
at least one bet.

There is also an equivalent graph-theoretic description. In the
direct product $K_3^{\times n}$, two vertices are adjacent exactly
when they differ in every coordinate, so $B_r$ is the open
neighborhood of $r$. Consequently~\cite{Brink2011},
\begin{equation}
 \tau_n
 =
 T(n,3)
 =
 \gamma_t\!\left(K_3^{\times n}\right),
\end{equation}
where $\gamma_t$ denotes the total domination number, i.e., the
minimum number of vertices whose open neighborhoods cover all
vertices of the graph.

For small $n$, substantially sharper information is available from
the covering literature. The known exact values are~\cite{Brink2011,KolevBaicheva2013}
\begin{equation}
 \bigl(T(n,3)\bigr)_{n=1}^{9}
 =
 (2,3,5,8,12,18,29,44,68).
 \label{eq:football-values-app}
\end{equation}

Using these exact values, \cref{tab:covering-small-app} compares the
winning-rectangle lower bound with the two-row-covering upper bound
for the first nine copy numbers.

\subsection{Asymptotic covering bound}
\label{app:covering-rate}

We now turn from the exact small-copy values of $T(n,3)$ to a general
bound that determines its asymptotic growth.

\begin{proof}[Proof of Lemma~\ref{lem:covering-bound-main}]
By the construction of Appendix~\ref{app:covering-finite}, each box
$B_r$ contains $2^n$ of the $3^n$ possible inputs. Any cover therefore
requires at least $3^n/2^n$ boxes. Since $T(n,3)$ is an integer,
\begin{equation}
 T(n,3)
 \geq
 \left\lceil
 \left(\frac32\right)^n
 \right\rceil.
\end{equation}

For the upper bound, choose $M$ excluded-row patterns independently
and uniformly from $\{0,1,2\}^n$. For any fixed input $x$, a random
pattern $r$ gives a box $B_r$ that covers $x$ with probability $p=\left(\frac23\right)^n$.
Hence, the probability that none of the $M$ chosen boxes covers $x$ is
$(1-p)^M$. Taking a union bound over all $3^n$ inputs gives
\begin{equation}
 \Prob[\text{some input is uncovered}]
 \leq
 3^n(1-p)^M
 <
 3^n e^{-Mp}.
 \label{eq:random-cover-probability-app}
\end{equation}

Choose
\begin{equation}
 M=M_n:=
 \left\lceil
 (n\ln3)\left(\frac32\right)^n
 \right\rceil.
\end{equation}
Since $M_np\geq n\ln3$, the last expression in
\cref{eq:random-cover-probability-app} is at most $1$. The strict
inequality therefore implies that the probability of leaving some
input uncovered is strictly less than $1$. Hence, there exists a
cover with at most $M_n$ boxes.

The random construction is used only to prove the existence of such a
cover. After removing any repeated patterns, we obtain a fixed cover
with at most $M_n$ distinct patterns. This proves the upper bound.
\end{proof}

\section{Explicit threshold comparison for MS}
\label{app:rao-jk}

Jain and Kundu combine Rao's threshold parallel-repetition bound
with a round-by-round message-guessing argument to obtain the scaling
$2^{-\Omega(\eta^3 n)+C}$ for interactive communication
(see Corollary~7.3 of Ref.~\cite{JainKundu2025}). Here, we make this comparison quantitative by
evaluating the explicit intermediate estimate in the proof of Rao's
concentration bound~\cite{Rao2011} and combining it with
the same message-guessing argument.

Under the uniform product input distribution and for
$0<\eta<1/9$, let $p_{n,\eta}^{(C)}$ denote the maximum probability
that a randomized fully interactive classical strategy, whose
communication cost is at most $C$ on every transcript, wins more
than an $(8/9+\eta)$ fraction of the $n$ copies. Define
\begin{equation}
 \Gamma_{\MS}^{\mathrm R}(\eta)
 :=
 \frac{-(\eta/100)^2
 \log_2\!\left[
 1-\dfrac{\eta/2}{8/9+3\eta/4}
 \right]}
 {4+(\eta/100)^2-\log_2(8/9+\eta/4)}.
 \label{eq:rao-rate-app}
\end{equation}
The resulting finite-$n$ bound is the following.

\begin{proposition}[Explicit threshold bound with communication]
\label{prop:rao-jk-app}
For every positive integer $n$, every $0<\eta<1/9$, and every
$C\geq0$,
\begin{equation}
 p_{n,\eta}^{(C)}
 \leq
 2^{C-\Gamma_{\MS}^{\mathrm R}(\eta)n+1}.
 \label{eq:rao-jk-app}
\end{equation}
Consequently, if a protocol wins more than an
$(8/9+\eta)$ fraction of the copies with probability at least
$\omega\in(0,1]$, then
\begin{equation}
 C
 \geq
 \Gamma_{\MS}^{\mathrm R}(\eta)n
 +
 \log_2\omega
 -1.
 \label{eq:rao-jk-cost-app}
\end{equation}
\end{proposition}

\begin{proof}
We first relate a bounded-communication strategy to one without
communication using the message-guessing argument of Jain and
Kundu~\cite{JainKundu2025}, adapted to our communication-cost
convention. Fix a deterministic protocol whose probability of
winning more than an $(8/9+\eta)$ fraction of the copies is $p$,
and suppose that every transcript has communication cost at most
$C$. Using shared randomness, the parties guess the message sent
at each internal node $v$, choosing uniformly from its $m(v)$
possible symbols. Since the guesses are shared, both parties can
follow the guessed transcript without communicating. If the guessed
transcript is inconsistent with a party's input, the corresponding
output may be chosen arbitrarily.

For an input whose original transcript is $\ell$, all guesses along
that transcript are correct with probability
\begin{equation}
 \prod_{v\in\operatorname{path}(\ell)}
 \frac1{m(v)}
 =
 2^{-L(\ell)}
 \geq
 2^{-C}.
 \label{eq:guess-probability-app}
\end{equation}
Whenever all guesses are correct, the no-communication strategy
reproduces the original protocol. Its probability of exceeding the
winning threshold is therefore at least $2^{-C}p$. Since
$p_{n,\eta}^{(0)}$ is the optimal threshold-winning probability
without communication,
\begin{equation}
 p
 \leq
 2^C p_{n,\eta}^{(0)}.
\end{equation}
The same argument applies to randomized protocols by first fixing the
protocol randomness and then averaging. Hence,
\begin{equation}
 p_{n,\eta}^{(C)}
 \leq
 2^C p_{n,\eta}^{(0)}.
 \label{eq:guess-factor-app}
\end{equation}

It remains to control the threshold-winning probability without
communication. For this, we use the explicit estimate from the proof
of Theorem~7 in Ref.~\cite{Rao2011}. For a two-player game of
classical value $1-\epsilon$ and answer length
$c:=\log_2|A|+\log_2|B|$, set
\begin{equation}
 s:=\frac{\eta}{100},
 \qquad
 t:=
 \frac{s^2n}
 {c+s^2-\log_2(1-\epsilon+\eta/4)}.
\end{equation}
With this choice, the proof in Ref.~\cite{Rao2011} gives
\begin{equation}
 p_{n,\eta}^{(0)}
 \leq
 2
 \left(
 1-\frac{\eta/2}{1-\epsilon+3\eta/4}
 \right)^{\lceil t\rceil}.
 \label{eq:rao-general-app}
\end{equation}
Since the quantity in parentheses lies between zero and one and
$\lceil t\rceil\geq t$, it follows that
\begin{equation}
 p_{n,\eta}^{(0)}
 \leq
2^{-\Gamma n+1},
\end{equation}
where
\begin{equation}
 \Gamma
 :=
 \frac{-s^2
 \log_2\!\left(
 1-\dfrac{\eta/2}{1-\epsilon+3\eta/4}
 \right)}
 {c+s^2-\log_2(1-\epsilon+\eta/4)}.
\end{equation}

For MS, $\epsilon=1/9$ and each party has four
possible answers, giving
$c=\log_2|A|+\log_2|B|=4$.
The coefficient $\Gamma$ therefore reduces precisely to
$\Gamma_{\MS}^{\mathrm R}(\eta)$ in
\cref{eq:rao-rate-app}. Combining this with
\cref{eq:guess-factor-app} proves \cref{eq:rao-jk-app}.
If the threshold event occurs with probability at least $\omega$,
then $\omega\leq p_{n,\eta}^{(C)}$, and rearranging
\cref{eq:rao-jk-app} gives \cref{eq:rao-jk-cost-app}.
\end{proof}

We now consider the event that all $n$ copies are won.
If this event occurs with probability at least $\omega$, then
\cref{eq:rao-jk-cost-app} holds for every $\eta<1/9$.
Taking $\eta\to1/9$ from below gives
\begin{equation}
 C_{\mathrm{wc},\mu_{\MS}^{\otimes n}}^{\leftrightarrow,(\omega)}
 (\MS^{\otimes n})
 \geq
 \pos{
 \gamma_{\MS}^{\mathrm R}n
 +
 \log_2\omega -1
 },
 \label{eq:rao-jk-all-copy-app}
\end{equation}
where
\begin{equation}
 \gamma_{\MS}^{\mathrm R}
 :=
 \lim_{\eta\to1/9^-}
 \Gamma_{\MS}^{\mathrm R}(\eta)
 \approx
 2.54031\times10^{-8}.
 \label{eq:rao-coefficient-app}
\end{equation}
For perfect winning, $\omega=1$, this becomes
\begin{equation}
 C_{\mathrm{wc},\mu_{\MS}^{\otimes n}}^{\leftrightarrow}
 (\MS^{\otimes n})
 \geq
 \pos{
 \gamma_{\MS}^{\mathrm R}n-1
 }.
 \label{eq:ms-rao-jk-main}
\end{equation}
The resulting lower bound becomes nonzero only when
$n\gtrsim3.937\times10^7$ and exceeds one bit only when
$n\gtrsim7.873\times10^7$.

\section{Hierarchy for input-average communication}
\label{app:hierarchy}

\subsection{Input-average cost beyond the one-bit polytope}
\label{app:PR}

Consider the three-input Popescu--Rohrlich correlation~\cite{Popescu_PRBox_1994,Barrett:PRA:2005}
\begin{equation}
 P_{\mathrm{PR},3}(a,b|x,y)
 =
 \begin{cases}
 \frac12, & a\oplus b=g(x,y),\\
 0, & \text{otherwise},
 \end{cases}
\end{equation}
where $a,b\in\{0,1\}$, $x,y\in\{0,1,2\}$, and
\begin{equation}
 g(x,y)
 =
 \begin{cases}
 1,
 &
  xy\in\{1,2\},
 \\
 0,
 &
 \text{otherwise}.
 \end{cases}
\end{equation}
Let $P_{\mathrm w}(a,b|x,y)=1/4$ for all $a,b,x,y$ and define
\begin{equation}
 \vec{P}_L
 :=
 \frac12\vec{P}_{\mathrm{PR},3}
 +
 \frac12\vec{P}_{\mathrm w}.
\end{equation}
The correlation $\vec{P}_L$ is local. To see this explicitly,
consider the following four deterministic sign assignments, ordered
as $(A_0,A_1,A_2;B_0,B_1,B_2)$:
\begin{equation}
\begin{aligned}
 &(1,1,1;\,1,-1,1),\qquad
 (1,1,-1;\,1,1,-1),\\
 &(1,-1,1;\,1,1,1),\qquad
 (1,-1,-1;\,-1,1,1).
\end{aligned}
\end{equation}
Here the signs encode the binary outputs according to
$A_x=(-1)^a$ and $B_y=(-1)^b$. Using these assignments with equal weights, together with their
simultaneous outcome relabelings, gives uniform marginals and
correlators
\begin{equation}
 E_{xy}(\vec{P}_L)
 :=
 \sum_{a,b}(-1)^{a\oplus b}P_L(a,b|x,y)
 =
 \frac12(-1)^{g(x,y)}.
\end{equation}
This reproduces the full correlation $\vec{P}_L$.

We now move from this local point toward
$\vec{P}_{\mathrm{PR},3}$ along the line
\begin{equation}
 \vec{P}_c
 :=
 c\vec{P}_{\mathrm{PR},3}
 +(1-c)\vec{P}_L,
 \qquad
 0<c\leq1.
 \label{eq:noisy-pr-app}
\end{equation}
We first determine when this line leaves the Alice-to-Bob one-bit
polytope~\cite{ZC:Entropy:2019}. Consider the correlator expression defined by
\begin{equation}
 M
 :=
 \begin{pmatrix}
 2&2&1\\
 2&-1&-2\\
 1&-2&2
 \end{pmatrix},
\end{equation}
through $\sum_{x,y}M_{xy}E_{xy}(\vec{P})$.
Every deterministic Alice-to-Bob one-bit correlation is specified by
signs $A_x\in\{+1,-1\}$, a message $m(x)\in\{0,1\}$, and receiver
signs $B_{y,k}\in\{+1,-1\}$. Optimizing over the receiver signs gives
\begin{align}
 \sum_{x,y}M_{xy}A_xB_{y,m(x)}
 &\leq
 \sum_{y,k}
 \left|
 \sum_{x:m(x)=k}M_{xy}A_x
 \right|
 \nonumber\\
 &\leq11.
 \label{eq:one-bit-pr-inequality-app}
\end{align}
The final maximum follows by exhausting the eight sign strings $A$
and the eight binary message functions $m$. By convexity, every
randomized Alice-to-Bob one-bit correlation therefore satisfies
\begin{equation}
 \sum_{x,y}M_{xy}E_{xy}(\vec{P})\leq11.
 \label{eq:one-bit-pr-bell-app}
\end{equation}

For the correlations on the line in
\cref{eq:noisy-pr-app}, we obtain
\begin{equation}
 E_{xy}(\vec{P}_c)
 =
 \frac{1+c}{2}(-1)^{g(x,y)}
 \end{equation}
 and
 \begin{equation}
 \sum_{x,y}M_{xy}(-1)^{g(x,y)}
 =15.
\end{equation}
Substituting these expressions into
\cref{eq:one-bit-pr-bell-app} gives $ c\leq\frac7{15}$.
Thus, $\vec{P}_c$ lies outside the Alice-to-Bob one-bit polytope
whenever $c>7/15$.

We next compare this one-bit restriction with the input-average
communication cost. The correlation $\vec{P}_{\mathrm{PR},3}$ can be
simulated by having Alice send her three-valued input to Bob, at a
communication cost of $\log_2 3$ bits, whereas $\vec{P}_L$ requires
no communication. Using shared randomness to choose between these
two simulations according to \cref{eq:noisy-pr-app} therefore gives
\begin{equation}
 \Cb_\mu^{\to}(\vec{P}_c)
 \leq
 c\log_2 3.
\end{equation}
Thus, for
\begin{equation}
 \frac7{15}<c<\frac1{\log_2 3},
\end{equation}
the correlation $\vec{P}_c$ cannot be decomposed using only
deterministic correlations with communication cost at most one bit,
while its one-way input-average communication cost is strictly below
one bit.

\subsection{Deterministic-cost formulation and hierarchy}

We now justify the deterministic-cost formulation and the hierarchy
introduced in \cref{sec:hierarchy-main}. The exact formulation uses
the full deterministic set, whereas each finite hierarchy level
requires only the cost classes retained at that level.

For one-way communication from Alice to Bob, a deterministic
correlation has the form
\begin{equation}
 D(a,b|x,y)
 =
 \delta_{a,A(x)}
 \delta_{b,B(x,y)}.
 \label{eq:one-way-det-app}
\end{equation}
Its communication cost has a direct coding interpretation. For each
$x$, define Bob's response function
\begin{equation}
 B_x:y\longmapsto B(x,y).
\end{equation}
If two inputs $x$ and $x'$ induce the same response function,
$B_x=B_{x'}$, then Bob need not distinguish between them. Alice
therefore only needs to communicate which response class contains her
input. Hence, $c_\mu^{\to}(\vec D)$ is obtained by minimizing the
average path cost of a prefix tree encoding these response classes.

For two-way communication, $c_\mu^{\leftrightarrow}(\vec D)$ is
instead obtained by minimizing over deterministic interactive
protocol trees. In both cases, the existence of an optimal finite
protocol follows from Lemma~\ref{lem:finite-tree-hierarchy-app}.

The exact decomposition in \cref{eq:exact-hierarchy-main} follows by
fixing the protocol randomness, which turns any randomized simulator
into a convex mixture of deterministic protocols. Conversely, any
such convex decomposition can be implemented using shared randomness
and minimum-cost deterministic protocols for its components.

We now prove the hierarchy.

\begin{proof}[Proof of Theorem~\ref{Thm:LowerBoundHierarchy}]
First, take any exact decomposition in
\cref{eq:exact-hierarchy-main} and retain only the terms with cost
at most $\kappa_\mathfrak{L}$. This gives a feasible point of
\cref{eq:relaxed-hierarchy-main}. Since every discarded term has
cost at least $\kappa_{\mathfrak{L}+1}$, charging the discarded
weight at $\kappa_{\mathfrak{L}+1}$ cannot increase the objective.
Thus, every hierarchy level gives a lower bound on
$\Cb_\mu^\sigma(\vec{P})$.

For monotonicity, take any feasible point at level $\mathfrak{L}+1$ and discard
the terms in $\Dset_{\mathfrak{L}+1}^\sigma$. The remaining terms are feasible
at level $\mathfrak{L}$. At level $\mathfrak{L}$, all discarded weight is charged at
$\kappa_{\mathfrak{L}+1}$, whereas at level $\mathfrak{L}+1$ the residual weight is charged
at $\kappa_{\mathfrak{L}+2}\geq\kappa_{\mathfrak{L}+1}$. Hence, the level-$\mathfrak{L}$ objective
cannot exceed the level-$(\mathfrak{L}+1)$ objective.

It remains to consider the final level. Let
$\{q_{\vec{D}}\}$ be an optimal solution at level $\mathfrak{L}=K-1$, and define
\begin{align}
 \rho
 &:=
 1-
 \sum_{k=0}^{K-1}
 \sum_{\vec{D}\in\Dset_k^\sigma}q_{\vec{D}},
 \\
 \vec{R}
 &:=
 \vec{P}-
 \sum_{k=0}^{K-1}
 \sum_{\vec{D}\in\Dset_k^\sigma}q_{\vec{D}}\vec{D}.
\end{align}
If $\rho=0$, this is already an exact decomposition. Otherwise,
$\vec{R}/\rho$ is a normalized correlation compatible with the
communication model $\sigma$. By the simulation constructions
preceding in Appendix~\ref{app:finite-simulation-cost}, it admits
a decomposition
\begin{equation}
 \frac{\vec{R}}{\rho}
 =
 \sum_{\vec{E}\in\Dset_K^\sigma}r_{\vec{E}}\vec{E}.
\end{equation}
Since every deterministic correlation has cost at most $\kappa_K$,
combining this decomposition with the retained terms gives an exact
decomposition of $\vec{P}$ with cost at most
\begin{equation}
 \sum_{k=0}^{K-1}
 \sum_{\vec{D}\in\Dset_k^\sigma}
 q_{\vec{D}}\kappa_k
 +
 \rho\kappa_K.
\end{equation}
This is precisely the final-level objective. Together with the
lower-bound property above, this proves equality at level $K-1$.
\end{proof}

It remains to identify the first level more explicitly. As discussed in
\cref{sec:hierarchy-main}, $\kappa_0=0$ and, in the scenarios
considered here, $\kappa_1=1$. Since $\Dset_0^\sigma$ is precisely the
set of local deterministic correlations, the first relaxation becomes
\begin{equation}
\begin{aligned}
 \Cb_{\mu}^{\sigma,(\leq\kappa_0)}(\vec{P})
 =
 \min_{\{q_{\vec{D}}\geq0:\vec{D}\in\Dset_0^\sigma\}}
 &\quad
 1-\sum_{\vec{D}\in\Dset_0^\sigma}q_{\vec{D}}
 \\
 \text{such that}
 &\quad
 \vec{P}\geq
 \sum_{\vec{D}\in\Dset_0^\sigma}q_{\vec{D}}\vec{D}.
\end{aligned}
\label{eq:first-level-app}
\end{equation}
For a no-signaling target $\vec{P}$, the residual in
\cref{eq:first-level-app} is a subnormalized no-signaling correlation.
Hence, the optimization is equivalent to finding the EPR2 decomposition~\cite{EPR92}
\begin{equation}
	 \vec{P} =  q_L\vec{P}_L+(1-q_L)\vec{P}_{\mathrm{NS}},
\end{equation}
with the largest possible $q_L$, or equivalently smallest possible $1-q_L$. Thus, the first hierarchy level coincides
with the EPR2 nonlocal content.

\end{document}